\documentclass[11pt,leqno]{amsart}%{article}

\usepackage{amssymb,amsfonts}%amsLatex
\usepackage{amsmath,amsthm,amsxtra}
\usepackage[all]{xy}
\usepackage{color}
\usepackage{verbatim}

\usepackage{enumerate}

\makeatletter
\@addtoreset{equation}{section}
\makeatother

\newtheorem{theorem}{Theorem}[section]

\newtheorem{lemma}[theorem]{Lemma}
\newtheorem{corollary}[theorem]{Corollary}
\newtheorem{proposition}[theorem]{Proposition}

\newtheorem*{lemma*}{Lemma}

\theoremstyle{remark}
\newtheorem{remark}[theorem]{Remark}

\newcommand\bigcheck[1]{#1 \raise1ex\hbox{$\hspace{-1ex}{}^\vee$}}
\newcommand\sucheck[1]{#1 \raise0.5ex\hbox{$\hspace{-1ex}{}^\vee$}}

\newcommand{\mc}[1]{{\mathcal #1}}
\newcommand{\mf}[1]{{\mathfrak #1}}
\newcommand{\mb}[1]{{\mathbb #1}}

\newcommand\tint{{\textstyle\int}}

\newcommand{\id}{{1 \mskip -5mu {\rm I}}}

\renewcommand{\tilde}{\widetilde}

\newcommand{\ad}{\mathop{\rm ad }}
\newcommand{\sign}{\mathop{\rm sign }}

\newcommand{\Mat}{\mathop{\rm Mat }}

\renewcommand{\ker}{\mathop{\rm Ker }}
\newcommand{\im}{\mathop{\rm Im }}

\newcommand{\Tr}{\mathop{\rm Tr }}

\newcommand{\rk}{\mathop{\rm rk }}

\newcommand{\Vect}{\mathop{\rm Vect }}

\newcommand{\var}{\mathop{\rm var }}

\definecolor{light}{gray}{.9}

\begin{document}

%%%%%%%%%%%%%%%%%%%%%%%%%%%%%%%%%%%%%%%%%%%%%%%%%%%%%%%%%%%%%%%%%%%%%%%%%%%%%%%%%%%%%%%%%%%%%%%%%%%%%%%%%%%%%%
%%%%%%%%%%%%%%% TITLE %%%%%%%%%%%%%%%%%%%%%%%%%%%%%%%%%%%%%%%%%%%%%%%%%%%%%%%%%%%%%%%%%%%%%%%%%%%%%%%%%%%%%%%
%%%%%%%%%%%%%%%%%%%%%%%%%%%%%%%%%%%%%%%%%%%%%%%%%%%%%%%%%%%%%%%%%%%%%%%%%%%%%%%%%%%%%%%%%%%%%%%%%%%%%%%%%%%%%%

\title{Essential variational Poisson cohomology}

\author{
Alberto De Sole,
Victor G. Kac
}

\thanks{Dipartimento di Matematica, Universit\`a di Roma ``La Sapienza'',
00185 Roma, Italy.
desole@mat.uniroma1.it
Supported by a PRIN grant and Fondi Ateneo, from the University of Rome.}

\thanks{Department of Mathematics, M.I.T.,
Cambridge, MA 02139, USA.
kac@math.mit.edu
Supported in part by an NSF grant.}

%%%%%%%%%%%%%%%%%%%%%%%%%%%%%%%%%%%%%%%%%%%%%%%%%%%%%%%%%%%%%%%%%%%%%%%%%%%%%%%%%%%%%%%%%%%%%%%%%%%%%%%%%%%%%%
%%%%%%%%%%%%%%% Abstract %%%%%%%%%%%%%%%%%%%%%%%%%%%%%%%%%%%%%%%%%%%%%%%%%%%%%%%%%%%%%%%%%%%%%%%%%%%%%%%%%%%%%%%
%%%%%%%%%%%%%%%%%%%%%%%%%%%%%%%%%%%%%%%%%%%%%%%%%%%%%%%%%%%%%%%%%%%%%%%%%%%%%%%%%%%%%%%%%%%%%%%%%%%%%%%%%%%%%%

\begin{abstract}
\noindent
In our recent paper \cite{DSK11} we computed the dimension of the variational Poisson cohomology $\mc H^\bullet_K(\mc V)$
for any quasiconstant coefficient $\ell\times\ell$ matrix differential operator $K$ of order $N$ with invertible leading coefficient,
provided that $\mc V$ is a normal algebra of differential functions over a linearly closed differential field.
In the present paper we show that, for $K$ skewadjoint, the $\mb Z$-graded Lie superalgebra $\mc H^\bullet_K(\mc V)$ is isomorphic
to the finite dimensional Lie superalgebra $\tilde H(N\ell,S)$.
We also prove that the subalgebra of ``essential'' variational Poisson cohomology, consisting
of classes vanishing on the Casimirs of $K$, is zero.
This vanishing result has applications to the theory of bi-Hamiltonian structures and their deformations.
At the end of the paper we consider also the translation invariant case.
\end{abstract}

\maketitle

%%%%%%%%%%%%%%%%%%%%%%%%%%%%%%%%%%%%%%%%%%%%%%%%%%%%%%%%%%%%%%%%%%%%%%%%%%%%%%%%%%%%%%%%%%%%%%%%%%%%%%%%%%%%%%
%%%%%%%%%%%%%%% Sect 1 %%%%%%%%%%%%%%%%%%%%%%%%%%%%%%%%%%%%%%%%%%%%%%%%%%%%%%%%%%%%%%%%%%%%%%%%%%%%%%%%%%%%%%%
%%%%%%%%%%%%%%%%%%%%%%%%%%%%%%%%%%%%%%%%%%%%%%%%%%%%%%%%%%%%%%%%%%%%%%%%%%%%%%%%%%%%%%%%%%%%%%%%%%%%%%%%%%%%%%

\section{Introduction}
\label{sec:intro}

The $\mb Z$-graded Lie superalgebra $W^{\var}(\Pi\mc V)=\bigoplus_{k=-1}^\infty W^{\var}_k$
of \emph{variational polyvector fields} is a very convenient framework for the theory of integrable
Hamiltonian PDE's.
This Lie superalgebra is associated to an algebra of differential functions $\mc V$,
which is an extension of the algebra of differential polynomials
$R_\ell=\mc F[u_i^{(n)}\,|\,i=1,\dots,\ell;\,n\in\mb Z_+]$
over a differential field $\mc F$ with the derivation $\partial$
extended to $R_\ell$ by $\partial u_i^{(n)}=u_i^{(n+1)}$.

The first three pieces, $W^{\var}_k$ for $k=-1,0,1$,
are identified with the most important objects in the theory of integrable systems:
First, $W^{\var}_{-1}=\Pi(\mc V/\partial\mc V)$,
where $\mc V/\partial\mc V$ is the space of \emph{Hamiltonian functions}
(or local functionals), and where $\Pi$ is just to remind that it should be considered
as an odd subspace of $W^{\var}(\Pi\mc V)$.
Second, $W^{\var}_0$ is the Lie algebra of \emph{evolutionary vector fields}
$$
X_P=\sum_{i=1}^\ell\sum_{n=0}^\infty(\partial^n P_i)\frac{\partial}{\partial u_i^{(n)}}\,,\,\,P\in\mc V^\ell\,,
$$
which we identify with $\mc V^\ell$.
Third, $W^{\var}_1$ is identified with the space of skewadjoint $\ell\times\ell$ matrix differential operators
over $\mc V$ endowed with odd parity.

For $\tint f,\,\tint g\in W^{\var}_{-1}$, $X,Y\in W^{\var}_0$, and $H=H(\partial)\in W^{\var}_1$,
the commutators are defined as follows
(as usual, $\tint$ denotes the canonical map $\mc V\to\mc V/\partial\mc V$):
\begin{eqnarray}
{[\tint f,\tint g]} &=& 0 \,, \label{0:eq2} \\
{[X,\tint f]} &=& \tint X(f) \,, \label{0:eq3} \\
{[X,Y]} &=& XY-YX \,, \label{0:eq4} \\
{[H,\tint f]} &=& H(\partial)\frac{\delta f}{\delta u} \,, \label{0:eq5} \\
{[X_P,H]} &=& X_P(H(\partial))-D_P(\partial)\circ H(\partial)-H(\partial)\circ D^*_P(\partial) \,. \label{0:eq6}
\end{eqnarray}
Here $\frac\delta{\delta u}$ is the variational derivative (see \eqref{eq:varder}),
$D_P$ is the Frechet derivative (see \eqref{frechet}),
and $D^*(\partial)$ denotes the matrix differential operator adjoint to $D(\partial)$.

The formula for the commutator of two elements $K,H$ of $W^{\var}_1$ (the so called
Schouten bracket) is more complicated (see \eqref{20110611:eq4},
but one needs only to know that conditions $[K,K]=0,\, [H,H]=0$ means that these matrix differential operators
are \emph{Hamiltonian}, and the condition $[K,H]=0$ means that they are \emph{compatible}.

There have been various various versions of the notion of variational polyvector fields,
but \cite{Kup80} is probably the earliest reference.

The basic notions of the theory of integrable Hamiltonian equations
can be easily described in terms of the Lie superalgebra $W^{\var}(\Pi\mc V)$.
Given a Hamiltonian operator $H$ and a Hamiltonian function $\tint h\in\mc V/\partial\mc V$,
the corresponding \emph{Hamiltonian equation} is
\begin{equation}\label{0:eq7}
\frac{du}{dt}=[H,\tint h]\,,\,\,u=(u_1,\dots,u_\ell)\,.
\end{equation}
One says that two Hamiltonian functions $\tint h_1$ and $\tint h_2$ are \emph{in involution} if
\begin{equation}\label{0:eq8}
[[H,\tint h_1],\tint h_2]=0\,.
\end{equation}
(Note that the LHS of \eqref{0:eq8} is skewasymmetric in $\tint h_1$ and $\tint h_2$,
since both are odd elements of the Lie superalgebra $W^{\var}(\Pi\mc V)$).
Any $\tint h_1$ which is in involution with $\tint h$ is called an \emph{integral of motion}
of the Hamiltonian equation \eqref{0:eq7},
and this equation is called \emph{integrable} if there exists an infinite dimensional
subspace $\Omega$ of $\mc V/\partial\mc V$ containing $\tint h$
such that all elements of $\Omega$ are in involution.
In this case we obtain a hierarchy of compatible integrable Hamiltonian equations,
labeled by elements $\omega\in\Omega$:
$$
\frac{du}{dt_\omega}=[H,\omega]\,.
$$

The basic device for proving integrability of a Hamiltonian equation is the so called
\emph{Lenard-Magri scheme}, proposed by Lenard in early 1970's (unpublished),
with an important imput by Magri \cite{Mag78}.
A survey of related results up to early 1990's can be found in \cite{Dor93},
and a discussion in terms of Poisson vertex algebras can be found in \cite{BDSK09}.

The Lenard-Magri scheme requires two compatible Hamiltonian operators $H$ and $K$
and a sequence of Hamiltonian functions $\tint h_n,\,n\in\mb Z_+$,
such that
\begin{equation}\label{0:eq8b}
[H,\tint h_n]=[K,\tint h_{n+1}]\,,\,\,n\in\mb Z_+\,.
\end{equation}
Then it is a trivial exercise in Lie superalgebra to show that all Hamiltonian functions
$\tint h_n$ are in involution
(hint: use the parenthetical remark after \eqref{0:eq8}).
Note to solve exercise one only uses the fact that $K,H$ lie in $W^{\var}_1$,
but in order to construct the sequence $\tint h_n,\,n\in\mb Z_+$,
one needs the Hamiltonian property of $H$ and $K$ and their compatibility.

The appropriate language here is the cohomological one.
Since $[K,K]=0$ and $K$ is an (odd) element of $W^{\var}_1$,
it follows that we have a cohomology complex
$$
\big(W^{\var}(\Pi\mc V)=\bigoplus_{k\geq-1}W^{\var}_k,\ad K\big)\,,
$$
called the variational Poisson cohomology complex.
As usual, let $\mc Z^\bullet_K(\mc V)=\bigoplus_{k\geq-1}\mc Z^-_K$
be the subalgebra of closed elements ($=\ker(\ad K)$),
and let $\mc B^\bullet_K(\mc V)=\bigoplus_{k\geq-1}\mc B^-_K$
be its ideal of exact elements ($=\im(\ad K)$).
Then the \emph{variational Poisson cohomology}
$$
\mc H^\bullet_K(\mc V)=\mc Z^\bullet_K(\mc V)\big/\mc B^\bullet_K(\mc V)
=\bigoplus_{k\geq-1}\mc H^k_K\,,
$$
is a $\mb Z$-graded Lie superalgebra.
(For usual polyvector fields the corresponding Poisson cohomology was introduced in \cite{Lic77}; cf. \cite{DSK11}).

Now we can try to find a solution to \eqref{0:eq8b} by induction on $n$ as follows
(see \cite{Kra88} and \cite{Olv87}).
Since $[K,H]=0$, we have, by the Jacobi identity:
\begin{equation}\label{0:eq11}
[K,[H,\tint h_n]]=-[H,[K,\tint h_n]]\,,
\end{equation}
hence, by the inductive assumption, the RHS of \eqref{0:eq11} is
$-[H,[H,\tint h_{n+1}]]$, which is zero since $[H,H]=0$ and $H$ is odd.
Thus, $[H,\tint h_n]\in\mc Z^0_K$.
To complete the $n$-th step of induction we need
that this element is exact, i.e. it equals $[H,\tint h_{n+1}]$ for some $\tint h_{n+1}$.
But in general we have
\begin{equation}\label{0:eq12}
[H,\tint h_n]=[K,\tint h_{n+1}]+z_{n+1}\,,
\end{equation}
where $z_{n+1}\in\mc Z^0_K$ only depends on the cohomology class in $\mc H^0_K$.

The best place to start the Lenard-Magri scheme is to take $\tint h_0=C_0\mc Z^{-1}_K$,
a \emph{central element} for $K$.
Then the first step of the Lenard-Magri scheme requires the existence of $\tint h_1$
such that
\begin{equation}\label{0:eq13}
[H,C_0]=[K,\tint h_{1}]\,.
\end{equation}
Taking bracket of both sides of \eqref{0:eq13} with arbitrary $C_1\in\mc Z^{-1}_K$,
we obtain
\begin{equation}\label{0:eq14}
[[H,C_0],C_1]=0\,.
\end{equation}
Thus, if we wish the Lenard-Magri scheme to work starting with an arbitrary central element $C_0$ for $K$,
the Hamiltonian operator $H$ (which lies in $\mc Z^1_K$), must satisfy \eqref{0:eq14} for any $C_0,C_1\in\mc Z^{-1}_K$.
In other words, $H$ must be ``essentially closed''.

It was remarked in \cite{DMS05} that condition \eqref{0:eq14} is an obstruction to triviality of deformations
of the Hamiltonian operator $K$,
which is, of course, another important reason to be interested in ``essential'' variational Poisson cohomology.

We define the subalgebra $\mc E\mc Z^\bullet_K(\mc V)=\bigoplus_{k\geq-1}\mc E\mc Z^k_K\subset\mc Z^\bullet_K(\mc V)$
of \emph{essentially closed elements}, by induction on $k\geq-1$, as follows:
$$
\mc E\mc Z^{-1}_K=0\,,\,\,
\mc E\mc Z^{k}_K=\big\{z\in\mc Z^k_K\,\big|\,[z,\mc Z^{-1}_K]\subset\mc E\mc Z^{k-1}\big\}
\,,\,\,k\in\mb Z_+\,.
$$
It is immediate to see that exact elements are essentially closed, and we define the \emph{essential variational
Poisson cohomology} as
$$
\mc E\mc H^\bullet_K(\mc V)=\mc E\mc Z^\bullet_K(\mc V)\big/\mc B^\bullet_K(\mc V)\,.
$$

The first main result of the present paper is Theorem \ref{20110612:thm},
which asserts that $\mc E\mc H^\bullet_K(\mc V)=0$,
provided that $K$ is an $\ell\times\ell$ matrix differential operator
of order $N$ with coefficients in $\Mat_{\ell\times\ell}(\mc F)$
and invertible leading coefficient,
that the differential field $\mc F$ is linearly closed,
and that the algebra of differential functions $\mc V$ is normal.
Recall that a differential field $\mc F$ is called \emph{linearly closed} \cite{DSK11}
if any linear inhomogenous (respectively homogenous) differential equation
of order greater than or equal to 1 with coefficients in $\mc F$
has a solution (resp. nonzero solution) in $\mc F$.

The proof of Theorem \ref{20110612:thm} relies on our previous paper \cite{DSK11},
where, under the same assumptions on $K$, $\mc F$ and $\mc V$,
we prove that $\dim_{\mc C}(\mc H^k_K)=\binom{N\ell}{k+2}$,
where $\mc C\subset\mc F$ is the subfield of constants,
and we constructed explicit representatives of cohomology classes.

In turn, Theorem \ref{20110612:thm} allows us to compute the Lie superalgebra structure
of $\mc H^\bullet_K(\mc V)$, which is our second main result.
Namely, Theorem \ref{20110601:thm} asserts that the $\mb Z$-graded Lie superalgebra
$\mc H^\bullet_K(\mc V)$ is isomorphic to the finite dimensional $\mb Z$-graded Lie superalgebra
$\tilde H(N\ell,S)$, of Hamiltonian vector fields over the Grassman superalgebra
in $N\ell$ indeterminates $\{\xi_i\}_{i=1}^{N\ell}$, with Poisson bracket $\{\xi_i,\xi_j\}=s_{ij}$,
divided by the central ideal $\mc C1$,
where $S=(s_{ij})$ is a nondegenerate symmetric $N\ell\times N\ell$ matrix over $\mc C$.

We hope that Theorem \ref{20110612:thm} will allow further progress
in the study of the Lenard-Magri scheme (work in progress).
First, it leads to classification of Hamiltonian operators $H$ compatible to $K$,
using techniques and results from \cite{DSKW10}.
Second, it shows that if the elements $z_{n+1}$ in \eqref{0:eq12} are essentially closed,
then they can be removed.

Also, of course, Theorem \ref{20110612:thm} shows that, if \eqref{0:eq14} holds
for a Hamiltonian operator obtained by a formal deformation of $K$,
then this formal deformation is trivial.
%
%As in \cite{DLZ06}, this should lead to ``quasitriviality'' of bi-Hamiltonian perturbations
%for an arbitrary Hamiltonian operator $K$ satisfying the assumptions of Theorem \ref{20110612:thm}.

In conclusion of the paper we discuss the other ``extreme'' -- the translation invariant case -- when $\mc F=\mc C$.
In this case,
we give an upper bound for the dimension of $\mc H^k_K$, for an arbitrary Hamiltonian operator $K$
with coefficients in $\Mat_{\ell\times\ell}(\mc C)$ and invertible leading coefficient,
and we show that this bound is sharp if and only if $K=K_1\partial$, where $K_1$
is a symmetric nondegenerate matrix over $\mc C$.
Since any Hamiltonian operator of hydrodymanic type can be brought,
by a change of variables, to this form,
our result generalizes the results of \cite{LZ11,LZ11pr} on $K$ of hydrodynamic type.
Furthermore, for such operators $K$
we also prove that the essential variational Poisson cohomology is trivial,
and we find a nice description of the $\mb Z$-graded Lie superalgebra $\mc H^\bullet_K$.

We are grateful to Tsinghua University and the Mathematical Sciences Center (MSC), Beijing, where this paper was written,
for their hospitality, and especially Youjin Zhang and Si-Qi Liu for enlightening lectures and discussions.
We also thank the Center of Mathematics and Theoretical Physics (CMTP), Rome, for continuing encouragement and support.

%%%%%%%%%%%%%%%%%%%%%%%%%%%%%%%%%%%%%%%%%%%%%%%%%%%%%%%%%%%%%%%%%%%%%%%%%%%%%%%%%%%%%%%%%%%%%%%%%%%%%%%%%%%%%%
%%%%%%%%%%%%%%% Sect before %%%%%%%%%%%%%%%%%%%%%%%%%%%%%%%%%%%%%%%%%%%%%%%%%%%%%%%%%%%%%%%%%%%%%%%%%%%%%%%%%%%%%%%
%%%%%%%%%%%%%%%%%%%%%%%%%%%%%%%%%%%%%%%%%%%%%%%%%%%%%%%%%%%%%%%%%%%%%%%%%%%%%%%%%%%%%%%%%%%%%%%%%%%%%%%%%%%%%%

\section{Transitive $\mb Z$-graded Lie superalgebras and prolongations}
\label{sec:before}

Recall \cite{GS64,Kac77} that a $\mb Z$-graded Lie superalgebra $\mf g=\bigoplus_{k\geq-1}\mf g_k$
is called \emph{transitive} if any $a\in\mf g_k,\,k\geq0$, such that $[a,\mf g_{-1}]=0$,
is zero.
Two equivalent definitions are as follows:
\begin{enumerate}[(i)]
\item
There are no nonzero ideals of $\mf g$ contained in $\bigoplus_{k\geq0}\mf g_k$.
\item
If $a\in\mf g_k$ is such that $[\dots[[a,C_0],C_1],\dots,C_k]=0$
for all $C_0,\dots,C_k\in\mf g_{-1}$, then $a=0$.
\end{enumerate}
If a $\mb Z$-graded Lie superalgebra $\mf g=\bigoplus_{k\geq-1}\mf g_k$ is transitive,
the Lie subalgebra $\mf g_0$ acts faithfully on $\mf g_{-1}$,
hence we have an embedding $\mf g_0\to gl(\mf g_{-1})$.

Given a Lie algebra $\mf g$ acting faithfully on a purely odd vector superspace $U$,
one calls a \emph{prolongation} of the pair $(U,\mf g)$
any transitive $\mb Z$-graded Lie superalgebra $\mf g=\bigoplus_{k\geq-1}\mf g_k$
such that $\mf g_{-1}=U$, $\mf g_0=\mf g$,
and the Lie bracket between $\mf g_0$ and $\mf g_{-1}$
is given by the action of $\mf g$ on $U$.
The \emph{full prolongation} of the pair $(U,\mf g)$
is a prolongation containing any other prolongation of $(U,\mf g)$.
It always exists and is unique.
%
%In fact, if $\mf p=\bigoplus_{k\geq-1}\mf p_k$ is a $\mb Z$-graded Lie superalgebra,
%with $\mf p_{-1}=U$, and with a Lie algebra homomorphism $\varphi:\,\mf p_0\to\mf g$
%compatible with the respective representations on $U$, i.e. $\varphi(A)v=[A,v]$ for all $A\in\mf p_0$ and $v\in U$,
%then there is a unique Lie superalgebra homomorphism $\varphi:\,\mf p\to FP(U,\mf g)$
%extending the identity map on $\mf p_{-1}=FP_{-1}(U,\mf g)=U$
%and the homomorphism $\varphi:\,\mf p_0\to FP_0(U,\mf g)=\mf g$.

\subsection{The $\mb Z$-graded Lie superalgebra $W(n)$}
\label{sec:before.1}

Let $\Lambda(n)$ be the Grassman superalgebra over the field $\mc C$
on odd generators $\xi_1,\dots,\xi_n$.
Let $W(n)$ be the Lie superalgebra of all derivations of the superalgebra $\Lambda(n)$,
with the following $\mb Z$-grading: for $k\geq-1$,
$W_k(n)$ is spanned by derivations of the form
$\xi_{i_1}\dots\xi_{i_{k+1}}\frac{\partial}{\partial \xi_j}$.
In particular, $W_{-1}(n)=\langle\frac{\partial}{\partial \xi_i}\rangle_{i=1}^n=\Pi\mc C^n$,
and $W_0(n)=\langle\xi_i\frac{\partial}{\partial \xi_j}\rangle_{i,j=1}^n\simeq gl(n)$.
It is easy to see that $W(n)$ is the full prolongation of $(\Pi\mc C^n,gl(n))$ \cite{Kac77}.
Consequently, any transitive $\mb Z$-graded Lie superalgebra $\mf g=\bigoplus_{k\geq-1}\mf g_k$,
with $\dim_{\mc C}\mf g_{-1}=n$, embeds in $W(n)$.

\subsection{The $\mb Z$-graded Lie superalgebra $\tilde H(n,S)$}
\label{sec:before.2}

Let $S=(s_{ij})_{i,j=1}^n$ be a symmetric $n\times n$ matrix over $\mc C$.
Consider the following subalgebra of the Lie algebra $gl(n)$:
\begin{equation}\label{20110628:eq1}
so(n,S)=\big\{A\in\Mat{}_{n\times n}(\mc C)\,\big|\,A^TS+SA=0\,,\,\,\Tr(A)=0\big\}\,.
\end{equation}

We endow the Grassman superalgebra $\Lambda(n)$ with a structure of a Poisson superalgebra by letting $\{\xi_i,\xi_j\}_S=s_{ij}$.
A closed formula for the Poisson bracket on $\Lambda(n)$ is
$$
 \{f,g\}_S=(-1)^{p(f)+1}\sum_{i,j=1}^n s_{ij} \frac{\partial f}{\partial \xi_i}\frac{\partial g}{\partial \xi_j}\,.
$$
We introduce a $\mb Z$-grading of the superspace $\Lambda(n)$
by letting $\deg(\xi_{i_1}\dots\xi_{i_s})=s-2$.
Note that this is a Lie superalgebra $\mb Z$ grading $\Lambda(n)=\bigoplus_{k=-2}^{n-2}\Lambda_k(n)$
(but it is not an associative superalgebra grading).
Note also that $\Lambda_{-2}(n)=\mc C1\subset\Lambda(n)$ is a central ideal
of this Lie superalgebra.
Hence $\Lambda(n)/\mc C1$ inherits the structure
of a $\mb Z$-graded Lie superalgebra of dimension $2^n-1$,
which we denote by $\tilde H(n,S)=\bigoplus_{k=-1}^{n-2}\tilde H_k(n,S)$.

The $-1$-st degree subspace is
$\tilde H_{-1}(n,S)=\langle\xi_i\rangle_{i=1}^n\simeq\Pi\mc C^n$,
and the $0$-th degree subspace $\tilde H_0(n,S)=\langle\xi_i\xi_j\rangle_{i,j=1}^n$
is a Lie subalgebra of dimension $\binom{n}{2}$.

Identifying $\tilde H_{-1}(n,S)$ with $\Pi\mc C^n$ (using the basis $\xi_i,\,i=1,\dots,n$)
and $\tilde H_0(n,S)$ with the space of skewsymmetric $n\times n$ matrices over $\mc C$
(via $\xi_i\xi_j\mapsto (E_{ij}-E_{ji})/2$),
the action of $\tilde H_0(n,S)$ on $\tilde H_{-1}(n,S)$ becomes:
$\{A,v\}_S=ASv$.
Note that, if $A$ is skewsymmetric, then $AS$ lies in $so(n,S)$.
Hence, we have a homomorphism of Lie superalgebras:
\begin{equation}\label{20110628:eq2}
\tilde H_{-1}(n,S)\oplus\tilde H_0(n,S)\to\Pi\mc C^n\oplus so(n,S)
\,\,,\,\,\,\,
(v,A)\mapsto (v,AS)
\,.
\end{equation}
\begin{lemma}\label{20110628:lem}
The map \eqref{20110628:eq2} is bijective if and only if $S$ has rank $n$ or $n-1$.
\end{lemma}
\begin{proof}
Clearly, if $S$ is nondegenerate, the map \eqref{20110628:eq2} is bijective.
Moreover, if $S$ has rank less than $n-1$, the map \eqref{20110628:eq2} is clearly not injective.
In the remaining case when $S$ has rank $n-1$,
we can assume it has the form
\begin{equation}\label{20110628:eq3}
S=\left(\begin{array}{cc}
0 & 0 \\
0 & T
\end{array}\right)\,,
\end{equation}
where $T$ is a nondegenerate symmetric $(n-1)\times(n-1)$ matrix.
In this case,
one immediately checks that the map \eqref{20110628:eq2} is injective.
Moreover,
$$
so(n,S)
=\Big\{\left(\begin{array}{cc} 0 & B^T \\ 0 & A \end{array}\right)
\,\Big|\, B\in\mc C^{\ell}\,,\,\, A\in so(n-1,T)\Big\}
\,.
$$
Hence, $\dim_{\mc C}so(n,S)=n-1+\binom{n-1}{2}=\binom{n}{2}=\dim_{\mc C}\tilde H_0(n,S)$.
\end{proof}
\begin{proposition}\label{20110622:prop}
If $S$ has rank $n$ or $n-1$, then $\tilde H(n,S)$
is the full prolongation of the pair $(\mc C^n,so(n,S))$.
\end{proposition}
\begin{proof}
For $S$ nondegenerate, the proof is can be found in \cite{Kac77}.
We reduce below the case $rk(S)=n-1$ to the case of nondegenerate $S$.
If $\rk(S)=\ell=n-1$, we can choose a basis $\langle\eta,\xi_1,\dots,\xi_\ell\rangle$,
such that the matrix $S$ is of the form \eqref{20110628:eq3}.
Define the map $\varphi_S:\,\tilde H(n,S)\to W(n)$, given by
\begin{equation}\label{20110622:eq1}
\begin{array}{l}
\displaystyle{
\varphi_S(f(\xi_1,\dots,\xi_\ell))
= \{f,\cdot\}_S=(-1)^{p(f)+1}\sum_{i,j=1}^\ell t_{ij} \frac{\partial f}{\partial \xi_i}\frac{\partial}{\partial \xi_j}\,,
} \\
\displaystyle{
\varphi_S(f(\xi_1,\dots,\xi_\ell)\eta)
= f(\xi_1,\dots,\xi_\ell)\frac{\partial}{\partial\eta}\,.
}
\end{array}
\end{equation}
It is easy to check that $\varphi_S$ is an injective homomorphism of $\mb Z$-graded Lie superalgebras.
Hence, we can identify $\tilde H(n,S)$ with its image in $W(n)$.

Since $\varphi_S(\tilde H_{-1}(n,S))=\Pi\mc C^n=W_{-1}(n)$,
the $\mb Z$-graded Lie superalgebra $\varphi_S(\tilde H(n,S))$ (hence $\tilde H(n,S)$) is transitive.
It remains to prove that it is the full prolongation of the pair $(\tilde H_{-1}(n,S),\tilde H_0(n,S))$.
For this, we will prove that, if
$$
X=f_0\frac{\partial}{\partial\eta}+\sum_{i=1}^\ell f_i\frac{\partial}{\partial\xi_i}\in W_k(n)\,,
$$
with $f_i\in\Lambda(n)$, homogenous polynomials of degree $k+1\geq2$, is such that
\begin{equation}\label{20110628:eq5}
\big[\frac\partial{\partial\eta},X\big]
\,\,,\,\,\,\,
\big[\frac\partial{\partial\xi_i},X\big]
\,\in \varphi_S(\tilde H_{k-1}(n,S))
\qquad\,\forall i=1,\dots\ell
\,,
\end{equation}
then $X\in\varphi_S(\tilde H_{k}(n,S))$.
Conditions \eqref{20110628:eq5} imply that all $f_0,\dots,f_\ell$ are polynomials in $\xi_1,\dots,\xi_\ell$ only,
and there exist $g_1,\dots,g_\ell$, polynomials in $\xi_1,\dots,\xi_\ell$,
such that
\begin{equation}\label{20110629:eq1}
\frac{\partial f_j}{\partial \xi_i}=(-1)^{p(g_i)+1}\sum_{k=1}^\ell t_{jk}\frac{\partial g_i}{\partial \xi_k}\,,
\end{equation}
for every $i,j\in\{1,\dots\ell\}$.
On the other hand, the condition that $X\in\varphi_S(\tilde H_{k}(n,S))$ means that
there exists $h$, a polynomial in $\xi_1,\dots,\xi_\ell$, such that
\begin{equation}\label{20110629:eq2}
f_i=
(-1)^{p(h)+1}\sum_{k=1}^\ell t_{ik}\frac{\partial h}{\partial g_k}\,.
\end{equation}
To conclude, we observe that conditions \eqref{20110629:eq1}
imply the existence of $h$ solving equation \eqref{20110629:eq2},
since $\tilde H(\ell,T)$ is a full prolongation.
\end{proof}

\begin{remark}
The notation $\tilde H(n,S)$ comes from the fact that,
if $S$ is nondegenerate, then the derived Lie superalgebra
$H(n,S)=\{\tilde H(n,S),\tilde H(n,S)\}=\bigoplus_{k=-1}^{n-3}\tilde H_k(n,S)$
has codimension 1 in $\tilde H(n,S)$, and it is simple for $n\geq4$.
\end{remark}

%%%%%%%%%%%%%%%%%%%%%%%%%%%%%%%%%%%%%%%%%%%%%%%%%%%%%%%%%%%%%%%%%%%%%%%%%%%%%%%%%%%%%%%%%%%%%%%%%%%%%%%%%%%%%%
%%%%%%%%%%%%%%% Sect 2 %%%%%%%%%%%%%%%%%%%%%%%%%%%%%%%%%%%%%%%%%%%%%%%%%%%%%%%%%%%%%%%%%%%%%%%%%%%%%%%%%%%%%%%
%%%%%%%%%%%%%%%%%%%%%%%%%%%%%%%%%%%%%%%%%%%%%%%%%%%%%%%%%%%%%%%%%%%%%%%%%%%%%%%%%%%%%%%%%%%%%%%%%%%%%%%%%%%%%%

\section{Variational Poisson cohomology}
\label{sec:2}

In this section we recall our results from \cite{DSK11} on the variational Poisson cohomology,
in the notation of the present paper.

\subsection{Algebras of differential functions}
\label{sec:2.0}

An \emph{algebra of differential functions} $\mc V$
in one independent variable $x$ and $\ell$ dependent variables $u_i$,
indexed by the set $I=\{1,\dots,\ell\}$,
is, by definition, a differential algebra
(i.e. a unital commutative associative algebra with a derivation $\partial$),
endowed with commuting derivations
$\frac{\partial}{\partial u_i^{(n)}}\,:\,\,\mc V\to\mc V$, for all $i\in I$ and $n\in\mb Z_+$,
such that, given $f\in\mc V$,
$\frac{\partial}{\partial u_i^{(n)}}f=0$ for all but finitely many $i\in I$ and $n\in\mb Z_+$,
and the following commutation rules with $\partial$ hold:
\begin{equation}\label{eq:comm_frac}
\Big[\frac{\partial}{\partial u_i^{(n)}} , \partial\Big] = \frac{\partial}{\partial u_i^{(n-1)}}\,,
\end{equation}
where the RHS is considered to be zero if $n=0$.
An equivalent way to write the identities \eqref{eq:comm_frac} is in terms of generating series:
\begin{equation}\label{eq:comm_frac_b}
\sum_{n\in\mb Z_+}z^n\frac{\partial}{\partial u_i^{(n)}}\circ \partial =
(z+\partial)\circ\sum_{n\in\mb Z_+}z^n\frac{\partial}{\partial u_i^{(n)}}\,.
\end{equation}
As usual we shall denote by $f\mapsto\tint f$ the canonical quotient map $\mc V\to\mc V/\partial\mc V$.

We call $\mc C=\ker(\partial)\subset\mc V$ the subalgebra of \emph{constant functions},
and we denote by $\mc F\subset\mc V$ the subalgebra of \emph{quasiconstant functions},
defined by
\begin{equation}\label{eq:4.2}
\mc F
=
\big\{f\in\mc V\,\big|\,\frac{\partial f}{\partial u_i^{(n)}}=0\,\,\forall i\in I,\,n\in\mb Z_+\big\}\,.
\end{equation}
It is not hard to show \cite{DSK11} that
$\mc C\subset\mc F$,
$\partial\mc F\subset\mc F$,
and
$\mc F\cap\partial\mc V=\partial\mc F$.
Throughout the paper we will assume that $\mc F$ is a field of characteristic zero,
hence so is $\mc C\subset\mc F$.
Unless otherwise specified, all vector spaces, as well as tensor products, direct sums, and Hom's,
will be considered over the field $\mc C$.

One says that $f\in\mc V$ has \emph{differential order} $n$ in the variable $u_i$
if $\frac{\partial f}{\partial u_i^{(n)}}\neq0$
and $\frac{\partial f}{\partial u_i^{(m)}}=0$ for all $m>n$.

The main example of an algebra of differential functions
is the ring of differential polynomials over a differential field $\mc F$,
$R_\ell\,=\,\mc F[u_i^{(n)}\,|\,i\in I,n\in\mb Z_+]$,
where $\partial(u_i^{(n)})=u_i^{(n+1)}$.
Other examples can be constructed starting from $R_\ell$
by taking a localization by some multiplicative subset $S$,
or an algebraic extension obtained by adding solutions of some polynomial equations,
or a differential extension obtained by adding solutions of some differential equations.

The \emph{variational derivative}
$\frac\delta{\delta u}:\,\mc V\to\mc V^{\ell}$
is defined by
\begin{equation}\label{eq:varder}
\frac{\delta f}{\delta u_i}\,:=\,\sum_{n\in\mb Z_+}(-\partial)^n\frac{\partial f}{\partial u_i^{(n)}}\,.
\end{equation}
It follows immediately from \eqref{eq:comm_frac_b} that $\partial\mc V\subset\ker \frac{\delta}{\delta u}$.

A \emph{vector field} is, by definition, a derivation of $\mc V$ of the form
\begin{equation}\label{2006_X}
X=\sum_{i\in I,n\in\mb{Z}_+} P_{i,n} \frac{\partial}{\partial u_i^{(n)}}\,\,,  \quad P_{i,n} \in \mc V\,.
\end{equation}
We denote by $\Vect(\mc V)$ the Lie algebra of all vector fields.
A vector field $X$ is called \emph{evolutionary} if $[\partial,X]=0$,
and we denote by $\Vect^\partial(\mc V)\subset\Vect(\mc V)$ the Lie subalgebra of all evolutionary vector fields.
By \eqref{eq:comm_frac}, a vector field $X$ is evolutionary
if and only if it has the form
\begin{equation}\label{2006_X2}
X_P=\sum_{i\in I,n\in\mb{Z}_+} (\partial^n P_i) \frac{\partial}{\partial u_i^{(n)}}\,,
\end{equation}
where $P=(P_i)_{i\in I}\in\mc V^\ell$, is called the \emph{characteristic} of $X_P$.

Given $P\in\mc V^{\ell}$, we denote by $D_P=\big((D_P)_{ij}(\partial)\big)_{i,j\in I}$
its \emph{Frechet derivative}, given by
\begin{equation}\label{frechet}
(D_P)_{ij}(\partial)=\sum_{n\in\mb Z_+}\frac{\partial P_i}{\partial u^{(n)}_j}\partial^n\,.
\end{equation}

Recall from \cite{BDSK09} that an algebra of differential functions $\mc V$
is called \emph{normal} if we have
$\frac\partial{\partial u_i^{(m)}}\big(\mc V_{m,i}\big)=\mc V_{m,i}$
for all $i\in I,m\in\mb Z_+$,
where we let
\begin{equation}\label{eq:july21_1}
\mc V_{m,i}\,:=\,\Big\{ f\in\mc V\,\Big|\,
\frac{\partial f}{\partial u_j^{(n)}}=0\,\,
\text{ if } (n,j)>(m,i) \text{ in lexicographic order }\Big\}\,.
\end{equation}
We also denote $\mc V_{m,0}=\mc V_{m-1,\ell}$, and $\mc V_{0,0}=\mc F$.

The algebra $R_\ell$ is obviously normal.
Moreover, any its extension $\mc V$ can be further extended to a normal algebra.
Conversely, it is proved in \cite{DSK09} that any normal algebra of differential functions $\mc V$
is automatically a differential algebra extension of $R_\ell$.
Throughout the paper we shall assume that $\mc V$ is an extension of $R_\ell$.

Recall also from \cite{DSK11} that
a differential field $\mc F$ is called \emph{linearly closed}
if any linear differential equation,
$$
a_nu^{(n)}+\dots+a_1u'+a_0u=b\,,
$$
with $n\geq0$, $a_0,\dots,a_n\in\mc F,\, a_n\neq0$,
has a solution in $\mc F$ for every $b\in\mc F$,
and it has a nonzero solution for $b=0$,
provided that $n\geq1$.

\subsection{The universal Lie superalgebra $W^{\var}(\Pi\mc V)$ of variational poly--vector fields}
\label{sec:2.1}

Recall the definition of the universal Lie superalgebra
of variational polyvector fields $W^{\var}(\Pi\mc V)$,
associated to the algebra of differential funtions $\mc V$ \cite{DSK11}.
We let
$$
W^{\var}(\Pi\mc V)=\bigoplus_{k=-1}^\infty W^{\var}_k\,,
$$
where
$W^{\var}_k$ is the superspace of parity $k\mod 2$
consisting of all \emph{skewsymmetric arrays}, i.e. arrays of polynomials
\begin{equation}\label{100518:eq2}
P=\big(P_{i_0,\dots,i_k}(\lambda_0,\dots,\lambda_k)\big)_{i_1,\dots,i_k\in I}\,,
\end{equation}
where $P_{i_0,\dots,i_k}(\lambda_0,\dots,\lambda_k)\in\mc V[\lambda_0,\dots,\lambda_k]/(\partial+\lambda_0+\dots+\lambda_k)$
are skewsymmetric with respect to simultaneous permutations of the variables $\lambda_0,\dots,\lambda_k$
and the indices $i_0,\dots,i_k$.
By $\mc V[\lambda_0,\dots,\lambda_k]/(\partial+\lambda_0+\dots+\lambda_k)$
we mean the quotient of the space $\mc V[\lambda_0,\dots,\lambda_k]$
by the image of the operator $\partial+\lambda_0+\dots+\lambda_k$.
Clearly, for $k=-1$ this space is $\mc V/\partial\mc V$
and, for $k\geq0$, we can identify it with
the algebra of polynomials $\mc V[\lambda_0,\dots,\lambda_{k-1}]$
by letting
$$
\lambda_k=-\lambda_0-\dots-\lambda_{k-1}-\partial\,,
$$
with $\partial$ acting from the left.
We then define the following $\mb Z$-graded Lie superalgebra bracket on $W^{\var}(\Pi\mc V)$.
For $P\in W^{\var}_h$ and $Q\in W^{\var}_{k-h}$, with $-1\leq h\leq k+1$, we let
$[P,Q]:=P\Box Q-(-1)^{h(k-h)}Q\Box P$,
where $P\Box Q\in W^{\var}_{k}$ is zero if $h=k-h=-1$, and otherwise it is given by
\begin{equation}\label{100418:eq2}
\begin{array}{l}
\displaystyle{
\big(P\Box Q\big)_{i_0,\dots,i_k}(\lambda_0,\dots,\lambda_k)
= \sum_{\sigma\in\mc S_{h,k}}
\sign(\sigma)
\sum_{j\in I,n\in\mb Z_+}
} \\
\displaystyle{
P_{j,i_{\sigma(k-h+1)},\dots,i_{\sigma(k)}}(\lambda_{\sigma(0)}+\dots+\lambda_{\sigma(k-h)}+\partial,\lambda_{\sigma(k-h+1)},\dots,\lambda_{\sigma(k)})_\to
} \\
\displaystyle{
(-\lambda_{\sigma(0)}-\dots-\lambda_{\sigma(k-h)}-\partial)^n
\frac{\partial}{\partial u_j^{(n)}}
Q_{i_{\sigma(0)},\dots,i_{\sigma(k-h)}}(\lambda_{\sigma(0)},\dots,\lambda_{\sigma(k-h)})
\,,
}
\end{array}
\end{equation}
where $\mc S_{h,k}$ denotes the set of $h$-\emph{shuffles} in the group $S_{k+1}=Perm\{0,\dots,k\}$,
i.e. the permutations $\sigma$ satisfying
$$
\sigma(0)<\cdots <\sigma(k-h)
\,\,,\,\,\,\,
\sigma(k-h+1)<\cdots< \sigma(k)\,.
$$
The arrow in \eqref{100418:eq2} means that $\partial$ should be moved to the right.
Note that, by the skewsymmetry conditions on $P$ and $Q$,
we can replace the sum over shuffles by the sum over the whole
permutation group $S_{k+1}$,
provided that we divide by $h!(k-h+1)!$.
It follows from Proposition 9.1 and the identification (9.22) in \cite{DSK11},
that the box product \eqref{100418:eq2} is well defined
and the corresponding commutator makes $W^{\var}(\Pi\mc V)$
into a $\mb Z$-graded Lie superalgebra.

\begin{remark}
In \cite{DSK11} we identified $W^{var}(\Pi\mc V)$ with the quotient space
$\Omega^\bullet(\mc V)=\tilde\Omega^\bullet(\mc V)/\partial\tilde\Omega^\bullet(\mc V)$,
where $\tilde\Omega^\bullet(\mc V)$ is the commutative associative unital superalgebra
freely generated over $\mc V$ by odd generators $\theta_i^{(m)}=\delta u_i^{(m)},\,i\in I,m\in\mb Z_+$,
and where $\partial:\,\tilde\Omega^\bullet(\mc V)\to\tilde\Omega^\bullet(\mc V)$
extends $\partial:\mc V\to\mc V$ to an even derivation such that $\partial\theta_i^{(m)}=\theta_i^{(m+1)}$.
This identification is given by mapping the array
$$
P=\Big(\sum_{m_0,\dots,m_k\in\mb Z_+}f^{m_0,\dots,m_k}_{i_0,\dots,i_k}\lambda_0^{m_0}\dots\lambda_k^{m_k}\Big)_{i_0,\dots,i_k\in I}\,\in W^{var}_k
$$
to the element
$$
\int\sum_{i_0,\dots,i_k\in I}\sum_{m_0,\dots,m_k\in\mb Z_+} f^{m_0,\dots,m_k}_{i_0,\dots,i_k}\theta_{i_0}^{(m_0)}\dots\theta_{i_k}^{(m_k)}
\,\in\Omega^{k+1}(\mc V)\,.
$$
(It is easy to see that this map is well defined and bijective.) Here $\tint$ denotes, as usual, the quotient map
$\tilde\Omega^\bullet(\mc V)\to\tilde\Omega^\bullet(\mc V)/\partial\tilde\Omega^\bullet(\mc V)=\Omega^\bullet(\mc V)$.
We extend the variational derivative to a map
$$
\frac{\delta}{\delta u_i}=\sum_{n\in\mb Z_+}(-\partial)^n\circ\frac{\partial}{\partial u_i^{(n)}}:\,\Omega^k(\mc V)\to\Omega^{k+1}(\mc V)
\,,
$$
by letting $\frac{\partial}{\partial u_i^{(n)}}$ acts on coefficients ($\in\mc V$).
Furthermore, we introduce the odd variational derivatives
$$
\frac{\delta}{\delta \theta_i}=\sum_{n\in\mb Z_+}(-\partial)^n\circ\frac{\partial}{\partial\theta_i^{(n)}}:\,\Omega^k(\mc V)\to\Omega^{k+1}(\mc V)
\,.
$$
Then the box product \eqref{100418:eq2} takes, under the identification $W^{var}(\Pi\mc V)\simeq\Omega^\bullet(\mc V)$,
the following simple form \cite{Get02}:
$$
P\Box Q=
\sum_{i\in I}\frac{\delta P}{\delta\theta_i}\frac{\delta Q}{\delta u_i}\,.
$$
\end{remark}

We describe explicitly the spaces $W^{\var}_k$ for $k=-1,0,1$.
Clearly, $W^{\var}_{-1}=\mc V/\partial\mc V$.
Also $W^{\var}_0=\mc V^\ell$
thanks to the obvious identification of $\mc V[\lambda]/(\partial+\lambda)$ with $\mc V$.
Finally,
the space $\mc V[\lambda,\mu]/(\partial+\lambda+\mu)$ is identified with
$\mc V[\lambda]\simeq\mc V[\partial]$,
by letting $\mu=-\partial$ moved to the left and $\lambda=\partial$ moved to the right.
Hence elements in $W^{\var}_1$ correspond to $\ell\times\ell$ matrix differential operators
over $\mc V$,
and the skewsymmetry condition for an element of $W^{\var}_1$
translates into the skewadjointness of the corresponding matrix differential operator
(i.e. to the condition $H^*_{ji}(\partial)=-H_{ij}(\partial)$,
where, as usual, for a differential operator $L(\partial)=\sum_nl_n\partial^n$,
its adjoint is $L^*(\partial)=\sum_n(-\partial)^n\circ l_n$).
In order to keep the same identification as in \cite{DSK11},
we associate to the array $P=\big(P_{ij}(\lambda,\mu)\big)_{i,j\in I}\in W^{\var}_1$,
the following skewadjoint $\ell\times\ell$ matrix differential operator
$H=\big(H_{ij}(\partial)\big)_{i,j\in I}$, where
\begin{equation}\label{20110608:eq3}
H_{ij}(\lambda)=P_{ji}(\lambda,-\lambda-\partial)\,,
\end{equation}
and $\partial$ acts from the left.

Next, we write some explicit formulas for the Lie brackets in $W^{\var}(\Pi\mc V)$.
Since $\mc S_{-1,k}=\emptyset$ and $\mc S_{k+1,k}=\{1\}$, we have,
for $\tint h\in \mc V/\partial\mc V=W^{\var}_{-1}$ and $Q\in W^{\var}_{k+1}$:
\begin{equation}\label{20110612:eqf2}
\begin{array}{l}
\displaystyle{
[\tint h,Q]_{i_0,\dots,i_k}(\lambda_0,\dots,\lambda_k)
=
(-1)^{k}[Q,\tint h]_{i_0,\dots,i_k}(\lambda_0,\dots,\lambda_k)
} \\
\displaystyle{
= (-1)^k
\sum_{j\in I}
Q_{j,i_{0},\dots,i_{k}}(\partial,\lambda_{0},\dots,\lambda_{k})_\to
\frac{\delta h}{\delta u_j}
\,,
}
\end{array}
\end{equation}
In particular, $[\tint h,\tint f]=0$ for $\tint f\in\mc V/\partial\mc V$.
For $Q\in\mc V^\ell=W^{\var}_{0}$ we have
\begin{equation}\label{20110613:eq3}
[Q,\tint h]=-[\tint h,Q]=\sum_{j\in I}\int Q_j\frac{\delta h}{\delta u_j}
=\tint X_Q(h)\,,
\end{equation}
where $X_Q$ is the evolutionary vector field with characteristics $Q$, defined in \eqref{2006_X2}.
Furthermore, for $H=\big(H_{ij}(\partial)\big)_{i,j\in I}\in W^{\var}_1$ (via the identification \eqref{20110608:eq3}), we have
\begin{equation}\label{20110608:eq1}
[H,\tint h]=H(\partial)\frac{\delta h}{\delta u}\in\mc V^\ell\,.
\end{equation}

Since $\mc S_{0,k}=\{1\}$ and $\mc S_{k,k}=\{(\alpha,0,\stackrel{\alpha}{\check{\dots}},k)\}_{\alpha=0}^k$, we have,
for $P\in \mc V^\ell=W^{\var}_0$ and $Q\in W^{\var}_k$,
$$
\begin{array}{l}
\displaystyle{
[P,Q]_{i_0,\dots,i_k}(\lambda_0,\dots,\lambda_k)
=
X_P\big(
Q_{i_0,\dots,i_k}(\lambda_0,\dots,\lambda_k)
\big)
}\\
\displaystyle{
-\sum_{\alpha=0}^k
\sum_{j\in I,n\in\mb Z_+}
Q_{i_0,\dots,\stackrel{\alpha}{\check{j}},\dots,i_k}(\lambda_0,\dots,\lambda_\alpha+\partial,\dots,\lambda_k)_\to
(-\lambda_\alpha-\partial)^n
\frac{\partial P_{i_\alpha}}{\partial u_j^{(n)}} \,.
}
\end{array}
$$
In particular, for $Q\in\mc V^\ell=W^{\var}_0$, we get the usual commutator of evolutionary vector fields:
$$
[P,Q]_i=X_P(Q_i)-X_Q(P_i)\,,
$$
while, for a skewadjoint $\ell\times\ell$ matrix differential operator $H(\partial)\in W^{\var}_1$,
we get
\begin{equation}\label{20110608:eq2}
\begin{array}{c}
\displaystyle{
[P,H](\partial)
=
X_P\big(H(\partial)\big)
- D_P(\partial)\circ H(\partial) - H(\partial)\circ D^*_P(\partial)\,,
}
\end{array}
\end{equation}
where, in the first term of the RHS, $X_P(H(\partial))$ denotes the $\ell\times\ell$ matrix differential operator
whose $(i,j)$ entry is obtained by applying $X_P$ to the coefficients of the differential operator $H_{ij}(\partial)$.
In the last two terms of the RHS of \eqref{20110608:eq2}, $D_P$ denotes the Frechet derivative of $P$, defined in \eqref{frechet},
and $D_P^*$ is its adjoint matrix differential operator.

Finally, we write equation \eqref{100418:eq2} in the case when $h=1$.
Since $\mc S_{1,k}=\{(0,\stackrel{\alpha}{\check{\dots}},k,\alpha)\}_{\alpha=0}^k$
and $\mc S_{k-1,k}=\{(\alpha,\beta,0,\stackrel{\alpha}{\check{\dots}}\stackrel{\beta}{\check{\dots}},k)\}_{0\leq\alpha<\beta\leq k}^k$,
we have, for a skewadjoint matrix differential operator
$H=\big(H_{ij}(\partial)\big)_{i,j\in I}\in W^{\var}_1$ (via the identificatino \eqref{20110608:eq3}) and for $P\in W^{\var}_{k-1}$:
\begin{equation}\label{20110611:eq1}
\begin{array}{l}
\displaystyle{
[H,P]_{i_0,\dots,i_k}(\lambda_0,\dots,\lambda_k)
%=[P,H]_{i_0,\dots,i_k}(\lambda_0,\dots,\lambda_k)
= (-1)^{k+1}
\sum_{j\in I,n\in\mb Z_+}
\sum_{\alpha=0}^k(-1)^{\alpha}
} \\
\displaystyle{
\Bigg(
\frac{
\partial
P_{i_0,\stackrel{\alpha}{\check{\dots}},i_k}(\lambda_0,\stackrel{\alpha}{\check{\dots}},\lambda_k)
}{\partial u_j^{(n)}}
(\lambda_\alpha+\partial)^n H_{j,i_\alpha}(\lambda_\alpha)
+ \sum_{\beta=\alpha+1}^k (-1)^{\beta}
} \\
\displaystyle{
\times
P_{j,i_0,\stackrel{\alpha}{\check{\dots}}\stackrel{\beta}{\check{\dots}},i_k}
(\lambda_\alpha+\lambda_\beta+\partial,\lambda_0,\stackrel{\alpha}{\check{\dots}}\stackrel{\beta}{\check{\dots}},\lambda_k)_\to
%} \\
%\displaystyle{
%\,\,\,\,\,\,\,
(-\lambda_\alpha-\lambda_\beta-\partial)^n
\frac{
\partial
H_{i_\beta,i_\alpha}(\lambda_\alpha)
}{\partial u_j^{(n)}}
\Bigg)\,.
}
\end{array}
\end{equation}
In particular, if $K=\big(K_{ij}(\partial)\big)_{i,j\in I}\in W^{\var}_1$, we have
$[K,H]=[H,K]=K\Box H+H\Box K$, where
\begin{equation}\label{20110611:eq4}
\begin{array}{l}
\displaystyle{
(K\Box H)_{i_0,i_1,i_2}(\lambda_0,\lambda_1,\lambda_2)
=
\sum_{j\in I,n\in\mb Z_+}
\Big(
\frac{\partial H_{i_0,i_1}(\lambda_1)}{\partial u_j^{(n)}}(\lambda_2+\partial)^n K_{j,i_2}(\lambda_2)
} \\
\displaystyle{
+\frac{\partial H_{i_1,i_2}(\lambda_2)}{\partial u_j^{(n)}}(\lambda_0+\partial)^n K_{j,i_0}(\lambda_0)
+\frac{\partial H_{i_2,i_0}(\lambda_0)}{\partial u_j^{(n)}}(\lambda_1+\partial)^n K_{j,i_1}(\lambda_1)
\Big)
\,.
}
\end{array}
\end{equation}
\begin{remark}\label{20110608:rem1}
Given a skewadjoint matrix differential operator $H=\big(H_{ij}(\partial)\big)$,
we can define the corresponding ``variational'' $\lambda$-brackets
$\{\cdot\,_\lambda\,\cdot\}_H:\,\mc V\times\mc V\to\mc V[\lambda]$,
given by the following formula (cf. \cite{DSK06}):
\begin{equation}\label{100502:eq2}
\{f_\lambda g\}=
\sum_{i,j\in I,m,n\in\mb Z_+}
\frac{\partial g}{\partial u_j^{(n)}}(\lambda+\partial)^n
H_{ji}(\lambda+\partial)
(-\lambda-\partial)^m\frac{\partial f}{\partial u_i^{(m)}}\,.
\end{equation}
One can write the above formulas in this language (cf. \cite{DSK11}).
\end{remark}

\begin{proposition}\label{20110613:prop2}
The $\mb Z$-graded Lie superalgebra $W^{\var}(\Pi\mc V)$ is transitive,
hence it is a prolongation of the pair $(\Pi\mc V/\partial\mc V,\Vect^\partial(\mc V))$.
\end{proposition}
\begin{proof}
First note that, if $H(\partial)$ is an $\ell\times\ell$ matrix differential operator such that
$H(\partial)\frac{\delta f}{\delta u}=0$ for every $f\in\mc V$, then $H(\partial)=0$ (cf. \cite{BDSK09}).
Indeed, if $H(\partial)$ has order $N$ and $H_{ij}(\partial)=\sum_{n=0}^Nh_{ij;n}\partial^n$ with $h_{ij;N}\neq0$,
then letting $f=\frac{(-1)^M}{2}(u_j^{(M)})^2$, we have
$\frac{\delta f}{\delta u_k}=\delta_{k,j}u_j^{(2M)}$
and, for $M$ sufficiently large, $\frac{\partial}{\partial u_j^{(2M+N)}}\big(H(\partial)\frac{\delta f}{\delta u}\big)_i=h_{ij;N}\neq0$
(here we are using the assumption that $\mc V$ contains $R_\ell$).
The claim follows immediately by this observation and equation \eqref{20110612:eqf2}.
\end{proof}

\subsection{The cohomology complex $(W^{\var}(\Pi\mc V),\delta_K)$}
\label{sec:2.3}

Let $K=\big(K_{ij}(\partial)\big)_{i,j\in I}\in W^{\var}_1$ be a \emph{Hamiltonian operator},
i.e. $K$ is skewadjoint and $[K,K]=0$.
Then $(\ad K)^2=0$, and we can consider
the associated \emph{variational Poisson cohomology complex} $(W^{\var}(\Pi\mc V),\ad K)$.
Let $\mc Z_K^\bullet(\mc V)=\bigoplus_{k=-1}^\infty\mc Z_K^k$,
where $\mc Z_K^k=\ker\big(\ad K\big|_{W^{\var}_k}\big)$,
and $\mc B_K^\bullet(\mc V)=\bigoplus_{k=-1}^\infty\mc B_K^k$,
where $\mc B_K^k=(\ad K)\big(W^{\var}_{k-1}\big)$.
Then $\mc Z_K^\bullet(\mc V)$ is a $\mb Z$-graded subalgebra
of the Lie superalgebra $W^{\var}(\Pi\mc V)$,
and $\mc B_K^\bullet(\mc V)$ is a $\mb Z$-graded ideal of $\mc Z_K^\bullet(\mc V)$.
Hence, the corresponding \emph{variational Poisson cohomology}
$$
\mc H_K^\bullet(\mc V)=\bigoplus_{k=-1}^\infty \mc H_K^k
\,\,,\,\,\,\,
\mc H_K^k=\mc Z_K^k\big/\mc B_K^k\,,
$$
is a $\mb Z$-graded Lie superalgebra.

In the special case when $K=\big(K_{ij}(\partial)\big)_{i,j\in I}$ has coefficients in $\mc F$,
which, as in \cite{DSK11}, we shall call a \emph{quasiconstant} $\ell\times\ell$ matrix differential operator,
formula \eqref{20110611:eq1} for the differential $\delta_K=\ad K$ becomes
for $P\in W^{\var}_{k-1},\,k\geq0$,
\begin{equation}\label{20110611:eq2}
\begin{array}{l}
\displaystyle{
(\delta_K P)_{i_0,\dots,i_k}(\lambda_0,\dots,\lambda_k)
} \\
\displaystyle{
= (-1)^{k+1}
\sum_{j\in I,n\in\mb Z_+}
\sum_{\alpha=0}^k(-1)^{\alpha}
\frac{
\partial
P_{i_0,\stackrel{\alpha}{\check{\dots}},i_k}(\lambda_0,\stackrel{\alpha}{\check{\dots}},\lambda_k)
}{\partial u_j^{(n)}}
(\lambda_\alpha+\partial)^n K_{j,i_\alpha}(\lambda_\alpha)
\,.
}
\end{array}
\end{equation}
In fact,
as shown in \cite[Prop.9.9]{DSK11}, if $K=\big(K_{ij}(\partial)\big)_{i,j\in I}$ is an arbitrary
quasiconstant $\ell\times\ell$ matrix differential operator
(not necessarily skewadjoint), then the same formula \eqref{20110611:eq2}
still gives a well defined linear map
$\delta_K:\,W^{\var}_{k-1}\to W^{\var}_k,\,k\geq0$, such that $\delta_K^2=0$.
Hence, we get a cohomology complex $(W^{\var}(\Pi\mc V),\delta_K)$.
As before, we denote
$\mc Z_K^k=\ker\big(\delta_K\big|_{W^{\var}_k}\big)\,,
\mc B_K^k=\delta_K\big(W^{\var}_{k-1}\big)$
and $\mc H_K^k=\mc Z_K^k\big/\mc B_K^k$.

For example,
$\mc H^{-1}_K=\mc Z^{-1}_K=
\big\{\tint f\in\mc V/\partial\mc V\,\Big|\,K^*(\partial)\frac{\delta f}{\delta u}=0\Big\}$,
which is called the set of \emph{central elements} (or \emph{Casimir elements}) of $K^*$.
Next, we have (see \cite{DSK11}):
$$
\mc B^0_K
=
\Big\{K^*(\partial)\frac{\delta f}{\delta u}
\Big\}_{f\in\mc V}
\,,\,\,
%$$
%and
%$$
\mc Z^0_K
=
\Big\{P\in\mc V^\ell
\,\Big|\,
D_P(\partial)\circ K(\partial)=K^*(\partial)\circ D_P^*(\partial)
\Big\}\,.
$$
Furthermore, given $P\in\mc V^\ell=W^{\var}_0$, the element $\delta_KP\in W^{\var}_1$,
under the identification \eqref{20110608:eq3} of $W^{\var}_1$ with the space of $\ell\times\ell$
skewadjoint matrix differential operators, coincides with
\begin{equation}\label{20110612:eq3}
\delta_KP=D_P(\partial)\circ K(\partial)-K^*(\partial)\circ D_P^*(\partial)\,.
\end{equation}
Hence, $\mc B^1_K=\big\{D_P(\partial)\circ K(\partial)-K^*(\partial)\circ D_P^*(\partial)\big\}_{P\in\mc V^\ell}$.
Finally, $\mc Z^1_K$ consists, under the same identification,
of the $\ell\times\ell$ skewadjoint matrix differential operators $H(\partial)$
for which the RHS of \eqref{20110611:eq4} is zero.

\begin{remark}
If $\tint f,\tint g\in\mc V/\partial\mc V$, we have
$[\tint f,\tint g]=0$ and
$$
[\delta_K\tint f,\tint g]-[\tint f,\delta_K\tint g]
=\tint \Big(-\frac{\delta g}{\delta u}K^*(\partial)\frac{\delta f}{\delta u}
-\frac{\delta f}{\delta u}K^*(\partial)\frac{\delta g}{\delta u}\Big)\,.
$$
Hence,
the differential $\delta_K$ in \eqref{20110611:eq2} is not an odd derivation unless $K(\partial)$
is skewadjoint.
In particular, the corresponding cohomology $\mc H^\bullet_K(\mc V)$ does not have
a natural structure of a Lie superalgebra unless $K(\partial)$ is a skewadjoint operator.
\end{remark}

\subsection{The variational Poisson cohomology $H(W^{\var}(\Pi\mc V),\delta_K)$ for a quasiconstant matrix differential operator $K(\partial)$}
\label{sec:2.4}

Let $\mc V$ be an algebra of differential functions extension of $R_\ell$, the algebra of differential polynomials
in the differential variables $u_1,\dots,u_\ell$ over a differential field $\mc F$.
Let $K=\big(K_{ij}(\partial)\big)_{i,j\in I}$ be a quasiconstant $\ell\times\ell$ matrix differential operator
of order $N$ (not necessarily skewadjoint).
For $k\geq-1$, we denote by $\mc A_K^k\subset W^{\var}_k$
the subset consisting of arrays of the form
\begin{equation}\label{20110612:eqf1}
\Big(\sum_{j\in I}
\big[
P_{j,i_0,\dots,i_k}(\lambda_0,\dots,\lambda_k)u_j\big]
\Big)_{i_0,\dots,i_k\in I}\,,
\end{equation}
where $[x]$ denotes the coset of $x\in\mc V[\lambda_0,\dots,\lambda_k]$ modulo $(\lambda_0+\dots+\lambda_k+\partial)\mc V[\lambda_0,\dots,\lambda_k]$,
satisfying the following properties.
For $j,i_0,\dots,i_k\in I$,
$P_{j,i_0,\dots,i_k}(\lambda_0,\dots,\lambda_k)$ are polynomials in $\lambda_0,\dots,\lambda_k$
with coefficients in $\mc F$ of degree at most $N-1$ in each variable $\lambda_i$,
skewsymmetric with respect to simultaneous permutations of the indices $i_0,\dots,i_k$,
and the variables $\lambda_0,\dots,\lambda_k$,
and satisfying the following condition:
\begin{equation}\label{20110611:eq3}
\begin{array}{r}
\displaystyle{
\sum_{\alpha=0}^{k+1}(-1)^\alpha\sum_{j\in I}
P_{j,i_0,\stackrel{\alpha}{\check{\dots}},i_{k+1}}(\lambda_0,\stackrel{\alpha}{\check{\dots}},\lambda_{k+1})K_{j i_\alpha}(\lambda_\alpha)
\equiv 0
} \\
\displaystyle{
\text{mod } (\lambda_0+\dots+\lambda_{k+1}+\partial)\mc F[\lambda_0,\dots,\lambda_{k+1}]\,.
}
\end{array}
\end{equation}

For example, $\mc A^{-1}_K$ consists of elements of the form
$\sum_{j\in I}\tint P_ju_j\in\mc V/\partial\mc V$, where $P\in\mc F^\ell$ solves
the equation
$$
K^*(\partial)P=0\,.
$$
In fact it is not hard to show that $\mc A^{-1}_K$ coincides with the set $\mc Z^{-1}_K$ of central elements of $K^*$
(see Lemma \ref{20110612:lemf} below).

Next, $\mc A^0_K$ consists of elements of the form
$\big(\sum_{j\in I}P^*_{ij}(\partial)u_j\big)_{i\in I}\in\mc V^\ell=W^{\var}_0$,
where $P=\big(P_{ij}(\partial)\big)_{i,j\in I}$
is a quasiconstant $\ell\times\ell$ matrix differential operator of order at most $N-1$,
solving the following equation:
\begin{equation}\label{110226:eq5}
K^*(\partial)\circ P(\partial)=P^*(\partial)\circ K(\partial)\,.
\end{equation}

The description of the set $\mc A^1_K$ is more complicated. Given a polynomial in two variables
$P(\lambda,\mu)=\sum_{m,n=0}^Nc_{mn}\lambda^m\mu^n\in\mc F[\lambda,\mu]$,
we denote $P^{*1}(\lambda,\mu)=\sum_{m,n=0}^N(-\lambda-\partial)^m c_{mn}\mu^n$,
and $P^{*2}(\lambda,\mu)=\sum_{m,n=0}^N(-\mu-\partial)^n c_{mn}\lambda^m$.
Then, under the identification of $W^{\var}_1$ with the space of skewadjoint $\ell\times\ell$ matrix differential operators
given by \eqref{20110608:eq3},
$\mc A^1_K$ consists of operators $H=\big(H_{ij}(\partial)\big)_{i,j\in I}$
of the form
$$
H_{ij}(\lambda)=-\sum_{k\in I}P^*_{kij}(\lambda+\partial,\lambda)u_k\,,
$$
where, for $i,j,k\in I$, $P_{kij}(\lambda,\mu)\in\mc F[\lambda,\mu]$
are polynomials of degree at most $N-1$ in each variable,
such that $P_{kij}(\lambda,\mu)=-P_{kji}(\mu,\lambda)$,
and such that
$$
\begin{array}{c}
\displaystyle{
\sum_{h\in I}\Big(
K^*_{ih}(\lambda+\mu+\partial) P_{hjk}(\lambda,\mu)
+P^{*2}_{hki}(\mu,\lambda+\mu+\partial) K_{hj}(\lambda)
} \\
\displaystyle{
+P^{*1}_{hij}(\lambda+\mu+\partial,\lambda) K_{hk}(\mu)
\Big)=0\,.
}
\end{array}
$$

Theorem 11.9 from \cite{DSK11} can be stated as follows:
\begin{theorem}\label{110213:thm}
Let $\mc V$ be a normal algebra of differential functions
in $\ell$ differential variables
over a linearly closed differential field $\mc F$,
and let $\mc C\subset\mc F$ be the subfield of constants.
Let $K(\partial)$ be a quasiconstant $\ell\times\ell$ matrix differential operator
of order $N$ with invertible leading coefficient $K_N\in\Mat_{\ell\times\ell}(\mc F)$.
Then
we have the following decomposition of $\mc Z^k_K$ in a direct sum of vector spaces over $\mc C$:
$$
\mc Z^k_K=\mc A^k_K\oplus \mc B^k_K\,.
$$
Hence, we have a canonical isomorphism $\mc H^k_K\simeq\mc A^k_K$.
Moreover, $\mc A^k_K$ (hence $\mc H^k_K$) is a vector space over $\mc C$ of dimension $\binom{N\ell}{k+2}$.
\end{theorem}

Recall that, if $K$ is a skewadjoint operator, then
$\mc H^\bullet_K(\mc V)=\bigoplus_{k\geq-1}\mc H^k_K$ is a Lie superalgebra
with consistent $\mb Z$-grading.
In Section \ref{sec:4} we will prove the following
\begin{theorem}\label{20110601:thm}
Let $\mc V$ be a normal algebra of differential functions,
over a linearly closed differential field $\mc F$.
Let $K(\partial)$ be a quasiconstant skewadjoint $\ell\times\ell$ matrix differential operator
of order $N$ with invertible leading coefficient $K_N\in\Mat_{\ell\times\ell}(\mc F)$.
Then the $\mb Z$-graded Lie superalgebra $\mc H^\bullet_K(\mc V)$ is isomorphic
to the $\mb Z$-graded Lie superalgebra $\tilde H(N\ell,S)$
constructed in Section \ref{sec:before.2},
where $S$ is the matrix, in some basis, of the nondegenerate symmetric bilinear form $\langle\,\cdot|\cdot\,\rangle^0_K$
constructed in Section \ref{sec:4.1}.
\end{theorem}
\begin{remark}
The subspace $\mc A^\bullet_K(\mc V)=\bigoplus_{k=-1}^\infty\mc A^k_K$ is NOT, in general,
a subalgebra of the Lie superalgebra $\mc Z^\bullet(\mc V)$.
We can enlarge it to be a subalgebra by letting $\tilde{\mc A}^k_K\subset\mc Z^k_K$
be the subset consisting of arrays of the form \eqref{20110612:eqf1}
where $P_{j,i_0,\dots,i_k}(\lambda_0,\dots,\lambda_k)$ are polynomials in $\lambda_0,\dots,\lambda_k$
with coefficients in $\mc F$ of arbitrary degree,
skewsymmetric with respect to simultaneous permutations of the indices $i_0,\dots,i_k$,
and the variables $\lambda_0,\dots,\lambda_k$, and satisfying condition \eqref{20110611:eq3}.
Then, clearly, $\mc A^\bullet_K(\mc V)\simeq\tilde{\mc A}^\bullet_K(\mc V)\big/\big(\tilde{\mc A}^\bullet_K(\mc V)\cap\mc B^\bullet_K(\mc V)\big)$.
For example, it is not hard to show that
$$
\tilde{\mc A}^0_K\cap\mc B^0_K = \big\{S(\partial)K(\partial)\,\big|S^*(\partial)=S(\partial)\big\}\,,
$$
so that, $\mc A^0_K$ is a Lie algebra, $\big\{S(\partial)K(\partial)\,\big|S^*(\partial)=S(\partial)\big\}$
is its ideal, and, by Theorem \ref{20110601:thm}, the quotient is isomorphic to the Lie algebra $so(N\ell)$.
\end{remark}

\begin{remark}\label{20110628:rem1}
If $N\leq1$, then $\mc A^\bullet_K(\mc V)$ is a subalgebra of the Lie superalgebra $\mc Z^\bullet_K(\mc V)$,
i.e. in this case the complex $(W^{var}(\Pi\mc V),\ad K)$ is formal (cf. \cite{Get02}).
However, this is not the case for $N>1$.
\end{remark}

%%%%%%%%%%%%%%%%%%%%%%%%%%%%%%%%%%%%%%%%%%%%%%%%%%%%%%%%%%%%%%%%%%%%%%%%%%%%%%%%%%%%%%%%%%%%%%%%%%%%%%%%%%%%%%
%%%%%%%%%%%%%%% Sect 3 %%%%%%%%%%%%%%%%%%%%%%%%%%%%%%%%%%%%%%%%%%%%%%%%%%%%%%%%%%%%%%%%%%%%%%%%%%%%%%%%%%%%%%%
%%%%%%%%%%%%%%%%%%%%%%%%%%%%%%%%%%%%%%%%%%%%%%%%%%%%%%%%%%%%%%%%%%%%%%%%%%%%%%%%%%%%%%%%%%%%%%%%%%%%%%%%%%%%%%

\section{Essential variational Poisson cohomology}
\label{sec:3}

In this section we introduce the subalgebra of essential variational Poisson cohomology
and we prove a vanishing theorem for this cohomology.

\subsection{The Casimir subalgebra $\mc Z^{-1}_K\subset\mc V/\partial\mc V$ and the essential subcomplex $\mc E W^{\var}(\Pi\mc V)$}
\label{sec:3.1}

Throughout this section we let $\mc V$ be an algebra of differential functions
in the variables $u_i,\,i\in I$,
and we denote, as usual, by $\mc F$ the subalgebra of quasiconstant,
and by $\mc C\subset \mc F$ the subalgebra of constants.
Let $K=\big(K_{ij}(\partial)\big)_{i,j\in I}$ be a Hamiltonian $\ell\times\ell$ matrix differential operator
with coefficients in $\mc V$.
In other words, we can view $K$ as an element of $W^{\var}_1$ such that $[K,K]=0$,
hence, we can consider the corresponding cohomology complex
$(W^{\var}(\Pi\mc V)=\bigoplus_{k\geq-1}W^{\var}_k,\ad K)$.
Recall from Section \ref{sec:2.3} that we have the $\mb Z$-graded subalgebra
$\mc Z^\bullet_K(\mc V)=\bigoplus_{k\geq-1}\mc Z^k_K$
of closed elements in $W^{\var}(\Pi\mc V)$,
and, inside it, the ideal of exact elements $\mc B^\bullet(\mc V)=\bigoplus_{k\geq-1}\mc B^k_K$.
The space $\mc Z^{-1}_K$ of central elements is, in this case,
\begin{equation}\label{20110612:eq1}
\mc Z^{-1}_K=\Big\{C\in\mc V/\partial\mc V\,\Big|\,[K,C]\,\Big(=K(\partial)\frac{\delta C}{\delta u}\Big)=0\Big\}
\,.
\end{equation}
We call an element $P\in W^{\var}_k$ \emph{essential} if the following condition holds:
\begin{equation}\label{20110612:eq2}
\big[\dots\big[[P,C_0],C_1\big],\dots,C_k\big]=0
\,,\,\,
\forall C_0,\dots,C_k\in\mc Z^{-1}_K\,.
\end{equation}
We denote by $\mc EW^{\var}_k\subset W^{\var}_k$ the subspace of essential elements.
For example, $\mc EW^{\var}_{-1}=0$
and $\mc EW^{\var}_0$ consists of elements $P\in\mc V^\ell$
such that $\tint P\frac{\delta C}{\delta u}=0$ for all central elements $C\in\mc Z^{-1}_K$.
Furthermore, $\mc EW^{\var}_1$ consists,
under the identification \eqref{20110608:eq3},
of skewadjoint $\ell\times\ell$ matrix differential operators $H(\partial)$,
such that
$$
\int \frac{\delta C_1}{\delta u}H(\partial)\frac{\delta C_2}{\delta u}=0
\,,\,\,
\forall C_1,C_2\in\mc Z^{-1}_K\,.
$$
Let $\mc EW^{\var}=\bigoplus_{k\geq-1}\mc EW^{\var}_k$.
This is a $\mb Z$-graded subspace of $W^{\var}(\Pi\mc V)$, depending on the operator $K(\partial)$.
Finally, denote by $\mc E\mc Z^\bullet_K(\mc V)=\bigoplus_{k\geq-1}\mc E\mc Z^k_K$ the $\mb Z$-graded
subspace of \emph{essentially closed elements}, i.e. $\mc E\mc Z^k_K=\mc Z^k_K\cap\mc EW^{\var}_k$.
\begin{proposition}\label{20110612:prop}
\begin{enumerate}[(a)]
\item
$\mc EW^{\var}$ is a $\mb Z$-graded subalgebra of the Lie superalgebra $W^{\var}(\Pi\mc V)$.
Consequently $\mc E\mc Z^\bullet_K(\mc V)$ is a $\mb Z$-graded subalgebra of $\mc EW^{\var}$.
\item
Exact elements are essentially closed, i.e. $B^\bullet_K(\mc V)\subset\mc E\mc Z^\bullet_K(\mc V)$,
hence they form a $\mb Z$-graded ideal  of the Lie superalgebra $\mc E\mc Z^\bullet_K(\mc V)$.
\end{enumerate}
\end{proposition}
\begin{proof}
Let $P\in\mc EW^{\var}_h$ and $Q\in\mc EW^{\var}_{k-h}$, with $0\leq h\leq k$,
and let $C_0,\dots,C_k\in\mc Z^{-1}_K$.
Using iteratively the Jacobi identity, we can express
$$
[\dots[[[P,Q],C_0],C_1],\dots,C_k]
$$
as a linear combination of the commutators of
the pairs of elements of the form
$$
[\dots[[P,C_{i_0}],C_{i_1}],\dots,C_{i_{s-1}}]
\,\,\text{ and }\,\,
[\dots[[Q,C_{i_s}],C_{i_{s+1}}],\dots,C_{i_k}]\,,
$$
where $s$ is either $h$ or $h+1$.
In the latter case the first element is zero since $P$ is essential,
while in the former case the second element is zero since $Q$ is essential.
Hence, $[P,Q]$ is essential.
The second claim of part (a) follows since $\mc E\mc Z^\bullet_K(\mc V)$
is the intersection of $\mc EW^{\var}$ and $\mc Z^\bullet_K(\mc V)$,
which are both $\mb Z$-graded subalgebra of $W^{\var}(\Pi\mc V)$.

For part (b), given the exact element $[K,P]$, where $P\in\mc EW^{\var}_{k-1}$,
and given $C_0,\cdots,C_k\in\mc Z^{-1}_K$,
we have, using again the Jacobi identity,
$$
[\dots[[[K,P],C_0],C_1],\dots,C_k]
=
[K,[\dots[[P,C_0],C_1],\dots,C_k]]
=0\,.
$$
\end{proof}
So, we define the \emph{essential variational Poisson cohomology} as
$$
\mc E\mc H^\bullet_K(\mc V)=\bigoplus_{k\geq-1}\mc E\mc H^k_K
\,\,,\text{ where }\,\,
\mc E\mc H^k_K=\mc E\mc Z^k_K/\mc B^k_K\,.
$$
Clearly, this is a $\mb Z$-graded subalgebra of the Lie superalgebra $\mc H^\bullet_K(\mc V)=H(W^{\var}(\Pi\mc V),\ad K)$.
\begin{remark}\label{20110612:rem1}
Let $H(\partial)$ be a Hamiltonian operator compatible with $K(\partial)$, i.e. $[K,H]=0$.
Suppose that the first step of the Lenard-Magri scheme always works,
namely for every central element $C\in\mc Z^{-1}_K$
there exists $\tint h\in\mc V/\partial\mc V$ such that $[H,C]=[K,\tint h]$.
Then $H$ is essentially closed.
Indeed, $[[H,C],C_1]=[[K,\tint h],C_1]=[\tint h,[K,C_1]]=0$ for every $C,C_1\in\mc Z^{-1}_K$.
This is one of the reasons for the name "essential", since only for the essentially closed operators $H$
the Lenard-Magri scheme may work.
Conversely, suppose $H(\partial)$ is an essentially closed Hamiltonian operator, i.e. $H(\partial)\in\mc E\mc Z^1_K$.
Then, for every central element $C\in\mc Z^{-1}_K$,
it is immediate to see that there exists $\tint h\in\mc V/\partial\mc V$
and $A\in\mc E\mc Z^0_K$ such that $[H,C]=[K,\tint h]+A$.
If the first essential variational Poisson cohomology is zero, we can choose $A$ to be zero,
which means that the first step in the Lenard-Magri scheme works.
\end{remark}

\subsection{Vanishing of the essential variational Poisson cohomology}
\label{sec:3.2}

In this section we prove the following
\begin{theorem}\label{20110612:thm}
If $\mc V$ be a normal algebra of differential functions in $\ell$ differential variables
over a linearly closed differential field $\mc F$,
and if $K(\partial)$ is a quasiconstant $\ell\times\ell$ matrix differential operator
of order $N$ with invertible leading coefficient $K_N\in\Mat_{\ell\times\ell}(\mc F)$,
then $\mc E\mc H^\bullet_K(\mc V)=0$.
\end{theorem}
In order to prove Theorem \ref{20110612:thm} we will need some preliminary lemmas.
\begin{lemma}\label{20110612:lemf}
Let $\mc V$ be an arbitrary algebra of differential functions.
Let $K(\partial):\,\mc V^\ell\to\mc V^\ell$ be a quasiconstant $\ell\times\ell$ matrix differential operator with
invertible leading coefficient $K_N\in\Mat_{\ell\times\ell}(\mc F)$
Then:
\begin{enumerate}[(a)]
\item
$\ker(K(\partial))=\ker\big(K(\partial)\big|_{\mc F^\ell}\big)$.
\item
The map $\frac{\delta}{\delta u}:\,\mc V/\partial\mc V\to\mc V^\ell$ restricts
to a surjective map $\frac{\delta}{\delta u}:\,\mc Z^{-1}_K\to\ker\big(K(\partial)\big|_{\mc F^\ell}\big)$.
\item
If, moreover, $\mc V$ is a normal algebra of differential functions
and $\partial:\,\mc F\to\mc F$ is surjective,
then we have a bijection
$\frac{\delta}{\delta u}:\,\mc Z^{-1}_K\stackrel{\sim}{\longrightarrow}\ker\big(K(\partial)\big|_{\mc F^\ell}\big)$.
\end{enumerate}
\end{lemma}
\begin{proof}
For part (a), we need to show that, if $F\in\mc V^\ell$ solves $K(\partial)F=0$, then $F\in\mc F^\ell$.
Suppose, by contradiction, that $F\notin\mc F^\ell$.
We may assume, without loss of generality, that $K_N=\id$,
and that the first coordinate $F_1$ has maximal differential order,
i.e. $F_1,\dots,F_\ell\in\mc V_{n,i}$ and $F_1\notin\mc V_{n,i-1}$, for some $i\in I,\,n\in\mb Z_+$.
Then $\frac{\partial}{\partial u_i^{(n+N)}}\big(K(\partial)F\big)_1=\frac{\partial F_1}{\partial u_i^{(n)}}\neq0$, a contradiction.
Next, we prove part (b). The inclusion $\frac{\delta}{\delta u}(\mc Z^{-1}_K)\subset\ker\big(K(\partial)\big|_{\mc F^\ell}\big)$
immediately follows from part (a).
Furthermore, if $P\in\ker\big(K(\partial)\big|_{\mc F^\ell}\big)$,
then $C=\tint\sum_iP_iu_i\in\mc Z^{-1}_K$ is such that $\frac{\delta C}{\delta u}=P$.
Hence, $\frac{\delta}{\delta u}\big(\mc Z^{-1}_K\big)=\ker\big(K(\partial)\big|_{\mc F^\ell}\big)$, as desired.
Finally, for part (c), if $\mc V$ is normal, we have by \cite[Prop.1.5]{BDSK09} that
$\ker\big(\frac{\delta}{\delta u}:\,\mc V/\partial\mc V\to\mc V^\ell)=\mc F/\partial\mc F$,
hence, if $\partial\mc F=\mc F$, we conclude that
$\frac{\delta}{\delta u}:\,\mc V/\partial\mc V\to\mc V^\ell$ is injective.
\end{proof}
To simplify notation, let $\mc Z:=\ker\big(K(\partial)\big)$.
Under the assumptions of Theorem \ref{20110612:thm}, by part (a) in Lemma \ref{20110612:lemf}, we have $\mc Z\subset\mc F^\ell$,
and by part (c) we have a bijection
\begin{equation}\label{20110614:eq1}
\frac{\delta}{\delta u}:\,\mc Z^{-1}_K\stackrel{\sim}{\longrightarrow}\mc Z\,,
\end{equation}
the inverse map being
$$
\mc Z\ni\,\,
F=
\left(\begin{array}{c}
f_1 \\ \vdots \\ f_\ell
\end{array}\right)
\mapsto\sum_i\tint f_iu_i
\,\,\in\mc Z^{-1}_K\,.
$$
\begin{lemma}\label{20110612:lem1}
If $F_1,\dots,F_{N\ell}$ are elements of $\mc F^\ell$, linearly independent over $\mc C$,
and satisfying a differential equation
\begin{equation}\label{20110612:eq0}
F^{(N)}=A_0F+A_1F'+\dots+A_{N-1}F^{(N-1)}\,,
\end{equation}
for some $A_0,\dots,A_{N-1}\in\Mat_{\ell\times\ell}(\mc F)$,
then the vectors
\begin{equation}\label{20110612:eq7}
G_1:=\left(\begin{array}{l}
F_1 \\
F_1' \\
\vdots \\
F_1^{(N-1)}
\end{array}\right)
\,,\,\,
\dots
\,,\,\,
G_{N\ell}:=\left(\begin{array}{l}
F_{N\ell} \\
F_{N\ell}' \\
\vdots \\
F_{N\ell}^{(N-1)}
\end{array}\right)
\,\in\mc F^{N\ell}
\end{equation}
are linearly independent over $\mc F$.
\end{lemma}
\begin{proof}
Suppose by contradiction that
\begin{equation}\label{20110612:eq4}
a_1G_1+a_2G_2+\dots+a_{N\ell}G_{N\ell}=0\,,
\end{equation}
is a nontrivial relation of linear dependence over $\mc F$.
We can assume, without loss of generality, that such relation has minimal number
of nonzero coefficients $a_1,\dots,a_{N\ell}\in\mc F$, and that $a_1=1$.
Note that equation \eqref{20110612:eq4} can be equivalently rewritten as the following system of equations in $\mc F^\ell$:
\begin{equation}\label{20110612:eq5}
\begin{array}{l}
a_1F_1+a_2F_2+\dots+a_{N\ell}F_{N\ell}=0\\
a_1F_1'+a_2F_2'+\dots+a_{N\ell}F_{N\ell}'=0\\
\,\,\,\,\,\,\,\,\,\,\,\,\,\,\,\,\,\,\dots \\
a_1F_1^{(N-1)}+a_2F_2^{(N-1)}+\dots+a_{N\ell}F_{N\ell}^{(N-1)}=0
\end{array}
\end{equation}
Applying $\partial$ to both sides of equation \eqref{20110612:eq4},
we get
\begin{equation}\label{20110612:eq6}
a_1G_1'+a_2G_2'+\dots+a_{N\ell}G_{N\ell}'
+a_1'G_1+ a_2'G_2+\dots+a_{N\ell}'G_{N\ell}
=0\,.
\end{equation}
The vector $a_1G_1'+a_2G_2'+\dots+a_{N\ell}G_{N\ell}'$ is an element of $\mc F^{N\ell}$
whose first $\ell$ coordinates are
$a_1F_1'+a_2F_2'+\dots+a_{N\ell}F_{N\ell}'$,
which are zero by the second equation in \eqref{20110612:eq5},
the second $\ell$ coordinates are
$a_1F_1^{(2)}+a_2F_2^{(2)}+\dots+a_{N\ell}F_{N\ell}^{(2)}$,
which are zero by the third equation in \eqref{20110612:eq5},
and so on, up to the last set of $\ell$ coordinates, which are,
by the equation \eqref{20110612:eq0},
$$
\begin{array}{l}
a_1F_1^{(N)}+a_2F_2^{(N)}+\dots+a_{N\ell}F_{N\ell}^{(N)} \\
=
A_0\big(a_1F_1+a_2F_2+\dots+a_{N\ell}F_{N\ell}\big)
+A_1\big(a_1F_1'+a_2F_2'+\dots+a_{N\ell}F_{N\ell}'\big)+ \\
\dots + A_{N-1}\big(a_1F_1^{(N-1)}+a_2F_2^{(N-1)}+\dots+a_{N\ell}F_{N\ell}^{(N-1)}\big)\,,
\end{array}
$$
which is zero again by the equations \eqref{20110612:eq5}.
Hence, equation \eqref{20110612:eq6} reduces to
$$
a_1'G_1 + a_2'G_2+\dots+a_{N\ell}'G_{N\ell}
=0\,,
$$
which, by the assumption that $a_1=1$ and the minimality assumption on the coefficients of linear dependence \eqref{20110612:eq4},
implies that all coefficients $a_1,\dots,a_{N\ell}$ are constant.
This, by the first equation in \eqref{20110612:eq5}, contradicts the assumption that $F_1,\dots,F_{N\ell}$
are linearly independent over $\mc C$.
\end{proof}
\begin{lemma}\label{20110612:lem2}
If $P(\partial)$ is a quasiconstant $m\times\ell$ ($m\geq1$) matrix differential operator
of order at most $N-1$ such that $P(\partial)F=0$ for every $F\in\mc Z=\ker\big(K(\partial)\big)$,
then $P(\partial)=0$.
\end{lemma}
\begin{proof}
Recall from \cite[Cor.A.3.7]{DSK11} that, if
$K(\partial)=K_0+K_1\partial+\dots+K_N\partial^N$,
with $K_i\in\Mat_{\ell\times\ell}(\mc F),\,i=0,\dots,N$ and $K_N$ invertible,
then the set of solutions in $\mc F^\ell$ of the homogeneous system
$K(\partial)F=0$ is a vector space over $\mc C$ of dimension $N\ell$.
Let $F_1,\dots,F_{N\ell}\in\mc F^\ell$ be a basis of this space.
Note that the equation $K(\partial)F=0$ has the form \eqref{20110612:eq0}
with $A_i=-K_N^{-1}K_i,\,i=0,\dots,N-1$.
Hence, by Lemma \ref{20110612:lem1},
all the vectors $G_1,\dots,G_{N\ell}$ in \eqref{20110612:eq7}
are linearly independent over $\mc F$,
i.e. the Wronskian matrix
$$
W=
\left(\begin{array}{llll}
F_1 & F_2 & \dots & F_{N\ell} \\
F_1' & F_2' & \dots & F_{N\ell}' \\
&& \dots & \\
F_1^{(N-1)} & F_2^{(N-1)} & \dots & F_{N\ell}^{(N-1)}
\end{array}\right)
$$
is nondegenerate.
By assumption $P(\partial)F_1=\dots =P(\partial)F_{N\ell}=0$.
Hence, letting $P(\partial)=P_0+P_1\partial+\dots+P_{N-1}\partial^{N-1}$,
where $P_i\in\Mat_{m\times\ell}(\mc F)$, we get
$$
\Big(P_0,P_1,\dots,P_{N-1}\Big)W=0\,,
$$
which, by the nondegeneracy of $W$, implies that $P_0=,\dots=P_{N-1}=0$.
\end{proof}
\begin{proof}[Proof of Theorem \ref{20110612:thm}]
Let $Q\in\mc A^k_K$.
Recalling Theorem \ref{110213:thm} and Proposition \ref{20110612:prop}(b), it suffices to show that,
if $Q$ is essential, then it is zero.
By the definition of $\mc A^k_K$, we have, in particular, that $Q$ is an array with entries
$$\begin{array}{r}
\displaystyle{
Q_{i_0,\dots,i_k}(\lambda_0,\dots,\lambda_k)
=
\sum_{j\in I}
P_{j,i_0,\dots,i_k}(\lambda_0,\dots,\lambda_k)u_j
} \\
\displaystyle{
\in\mc V[\lambda_0,\dots,\lambda_k]/(\partial+\lambda_0+\dots+\lambda_k)\mc V[\lambda_0,\dots,\lambda_k]\,,
}
\end{array}
$$
for some polynomials $P_{j,i_0,\dots,i_k}(\lambda_0,\dots,\lambda_k)\in\mc F[\lambda_0,\dots,\lambda_k]$
of degree at most $N-1$ in each variable $\lambda_i$.
Recalling formula \eqref{20110612:eqf2}, we have,
for arbitrary $C_0,\dots,C_k\in\mc V/\partial\mc V$,
\begin{equation}\label{20110612:eqf3}
[\dots[[Q,C_0],C_1],\dots,C_k]
=
\sum_{j,i_0,\dots,i_k\in I}
\int u_j P_{j,i_0,\dots,i_k}(\partial_0,\dots,\partial_k)
\frac{\delta C_0}{\delta u_{i_0}}\dots\frac{\delta C_k}{\delta u_{i_k}}\,,
\end{equation}
where $\partial_s$ means $\partial$ acting on $\frac{\delta C_s}{\delta u_{i_s}}$.
Hence, if $Q$ is essential,
\eqref{20110612:eqf3} is zero for all $C_0,\dots,C_k\in\mc Z^{-1}_K$.
By Lemma \ref{20110612:lemf}, we thus have
$$
\sum_{j,i_0,\dots,i_k\in I}
\int u_j P_{j,i_0,\dots,i_k}(\partial_0,\dots,\partial_k)
F_0\dots F_k=0\,,
$$
for all $F_0,\dots,F_k\in\ker\big(K(\partial)\big|_{\mc F^\ell}\big)$.
Since all coefficients of the $P_{j,i_0,\dots,i_k}$'s and all entries of the $F_i$'s are quasiconstant,
the above equation is equivalent to
$$
\sum_{i_0,\dots,i_k\in I}
P_{j,i_0,\dots,i_k}(\partial_0,\dots,\partial_k)
F_0\dots F_k=0\,,\,\,\forall j\in I\,.
$$
Applying Lemma \ref{20110612:lem2} iteratively to each factor,
we conclude that the polynomials $P_{j,i_0,\dots,i_k}(\lambda_0,\dots,\lambda_k)$
are zero.
\end{proof}

\begin{remark}
By Remark \ref{20110612:rem1}, from the point of view of applicability of the Lenard-Magri scheme for a bi-Hamiltonian pair $(H,K)$,
we should consider only essentially closed Hamiltonian operators $H(\partial)$.
Moreover, by Theorem \ref{20110612:thm}, if $K(\partial)$ is a quasiconstant matrix differential operator with invertible leading coefficient,
an essentially closed $H(\partial)$ must be exact, namely, recalling equation \eqref{20110612:eq3}, it must have the form
$$
H(\partial)=
D_P(\partial)\circ K(\partial) + K(\partial)\circ D^*_P(\partial)\,,
$$
for some $P\in\mc V^\ell$,
and two such $P$'s differ by an element of the form $K(\partial)\frac{\delta f}{\delta u}$ for some $\tint f\in\mc V/\partial\mc V$.
\end{remark}
\begin{corollary}\label{20110613:cor}
Under the assumptions of Theorem \ref{20110612:thm},
the $\mb Z$-graded Lie superalgebra $\mc H^\bullet_K(\mc V)$ is transitive.
\end{corollary}
\begin{proof}
By Theorem \ref{20110612:thm}, if $P\in\mc H^k_K$ is such that
$[\dots[[P,C_0],C_1],\dots,C_k]=0$ for every $C_0,\dots,C_k\in\mc Z^{-1}_K=\mc H^{-1}_K$,
then $P=0$.
This, by definition, means that $\mc H^\bullet_K(\mc V)$ is transitive.
\end{proof}

%%%%%%%%%%%%%%%%%%%%%%%%%%%%%%%%%%%%%%%%%%%%%%%%%%%%%%%%%%%%%%%%%%%%%%%%%%%%%%%%%%%%%%%%%%%%%%%%%%%%%%%%%%%%%%
%%%%%%%%%%%%%%% Sect 4 %%%%%%%%%%%%%%%%%%%%%%%%%%%%%%%%%%%%%%%%%%%%%%%%%%%%%%%%%%%%%%%%%%%%%%%%%%%%%%%%%%%%%%%
%%%%%%%%%%%%%%%%%%%%%%%%%%%%%%%%%%%%%%%%%%%%%%%%%%%%%%%%%%%%%%%%%%%%%%%%%%%%%%%%%%%%%%%%%%%%%%%%%%%%%%%%%%%%%%

\section{Isomorphism of $\mb Z$-graded Lie superalgebras $\mc H^\bullet_K(\mc V)\simeq\tilde H(N\ell,S)$}\label{sec:4}

In this section
we introduce an inner product $\langle\,\cdot|\cdot\,\rangle_K:\,\mc F^\ell\times\mc F^\ell\to\mc F$
associated to an $\ell\times\ell$ matrix differential operator $K=\big(K_{ij}(\partial)\big)_{i,j\in I}$,
which is used to prove Theorem \ref{20110601:thm}.

\subsection{The inner product associated to $K$}\label{sec:4.1}

Let $\mc F$ be a differential algebra with derivation $\partial$, and denote by $\mc C$ the subalgebra of constants.
As usual, we denote by $\cdot$ the standard inner product on $\mc F^\ell$, i.e.
$F\cdot G=\sum_{i\in I}F_iG_i\in\mc V$ for $F,G\in\mc V^\ell$, where, as before, $I=\{1,\dots,\ell\}$.

Consider the algebra of polynomials in two variables $\mc F[\lambda,\mu]$.
Clearly, the map $\lambda+\mu+\partial:\,\mc F[\lambda,\mu]\to\mc F[\lambda,\mu]$ is injective.
Hence, given $P(\lambda,\mu)\in(\lambda+\mu+\partial)\mc F[\lambda,\mu]$,
there is a unique preimage of this map in $\mc F[\lambda,\mu]$, that we denote by
$(\lambda+\mu+\partial)^{-1}P(\lambda,\mu)\in\mc F[\lambda,\mu]$.

Let now $K(\partial)=\big(K_{ij}(\partial)\big)_{i,j\in I}$ be an arbitrary $\ell\times\ell$ matrix differential operator over $\mc F$.
We expand its matrix entries as
\begin{equation}\label{20110613:eq1}
K_{ij}(\lambda)
=\sum_{n=0}^NK_{ij;n}\lambda^n
\,\,,\,\,\,\,
K_{ij;n}\in\mc F
\,.
\end{equation}
The adjoint operator is $K^*(\partial)$, with entries
\begin{equation}\label{20110613:eq1b}
K^*_{ij}(\lambda)=K_{ji}(-\lambda-\partial)
=\sum_{n=0}^N
(-\lambda-\partial)^n K_{ji;n}\,.
\end{equation}
It follows from the expansions \eqref{20110613:eq1} and \eqref{20110613:eq1b}
that, for every $i,j\in I$, the polynomial $K_{ij}(\mu)-K^*_{ji}(\lambda)$ lies
in the image of $\lambda+\mu+\partial$, so that we can consider the polynomial
\begin{equation}\label{20110617:eq1}
(\lambda+\mu+\partial)^{-1}\big(K_{ij}(\mu)-K^*_{ji}(\lambda)\big)\,\in\mc F[\lambda,\mu]\,.
\end{equation}

Next, for a polynomial $P(\lambda,\mu)=\sum_{m,n=0}^Np_{mn}\lambda^m\mu^n \in\mc F[\lambda,\mu]$,
we use the following notation
\begin{equation}\label{20110617:eq2}
P(\lambda,\mu)\big(|_{\lambda=\partial}f\big)\big(|_{\mu=\partial}g\big)
:=\sum_{m,n=0}^Np_{mn}(\partial^mf)(\partial^ng)\,.
\big({}_{\lambda}|_{=\partial}f\big)
\end{equation}

Based on the observation \eqref{20110617:eq1}, and using the notation in \eqref{20110617:eq2},
we define the following inner product
$\langle\,\cdot|\cdot\,\rangle_K:\,\mc F^\ell\times\mc F^\ell\to\mc F$,
associated to $K=\big(K_{ij}(\partial)\big)_{i,j\in I}\in\Mat_{\ell\times\ell}(\mc F[\partial])$:
\begin{equation}\label{20110617:eq3}
\langle F|G\rangle_K
=\sum_{i,j\in I}
(\lambda+\mu+\partial)^{-1}\big(K_{ij}(\mu)-K^*_{ji}(\lambda)\big)
\big(|_{\lambda=\partial}F_i\big)\big(|_{\mu=\partial}G_j\big)
\,.
\end{equation}

It is not hard to write an explicit formula for $\langle F|G\rangle_K$,
using the expansion \eqref{20110613:eq1} for $K_{ij}(\lambda)$:
\begin{equation}\label{20110613:eq2}
\langle F|G\rangle_K
= \sum_{i,j\in I}\sum_{n=0}^N\sum_{m=0}^{n-1} \binom{n}{m}
(-\partial)^{n-1-m} (F_i K_{ij;n} \partial^m G_j) \,.
\end{equation}

\begin{lemma}\label{20110617:lem1}
For every $F,G\in\mc V^\ell$, we have
$$
\partial\langle F|G\rangle_K=F\cdot K(\partial)G-G\cdot K^*(\partial)F\,.
$$
\end{lemma}
\begin{proof}
It immediately follows from the definition \eqref{20110617:eq3} of $\langle F|G\rangle_K$.
\end{proof}

\begin{lemma}\label{20110617:lem2}
For every $K(\partial)\in\Mat_{\ell\times\ell}(\mc F[\partial])$ and $F,G\in\mc F^\ell$, we have
$$
\langle G|F\rangle_{K^*}=-\langle F|G\rangle_K\,.
$$
In particular, the inner product $\langle\,\cdot|\cdot\,\rangle_K$ is symmetric (respectively skewsymmetric)
if $K$ is skewadjoint (resp. selfadjoint).
\end{lemma}
\begin{proof}
By equation \eqref{20110617:eq3} we have
$$
\begin{array}{l}
\displaystyle{
\langle G|F\rangle_{K^*}
=\sum_{i,j\in I}
(\lambda+\mu+\partial)^{-1}\big(K^*_{ij}(\mu)-K_{ji}(\lambda)\big)
\big(|_{\lambda=\partial}G_i\big)\big(|_{\mu=\partial}F_j\big)
} \\
\displaystyle{
=-\sum_{i,j\in I}
(\lambda+\mu+\partial)^{-1}\big(K_{ij}(\mu)-K^*_{ji}(\lambda)\big)
\big(|_{\lambda=\partial}F_i\big)\big(|_{\mu=\partial}G_j\big)
=-\langle F|G\rangle_K
\,. } \\
\end{array}
$$
\end{proof}

Following the notation of the previous sections, we let
$\mc Z=\ker\big(K(\partial)\big)\subset\mc F^\ell$.
Clearly, $\mc Z$ is a submodule of the $\mc C$-module $\mc F^\ell$.
\begin{lemma}\label{20110617:lem3}
If $K(\partial)\in\Mat_{\ell\times\ell}(\mc F[\partial])$ is skewadjoint, then
$\langle F|G\rangle_K\in\mc C$ for every $F,G\in\mc Z$
\end{lemma}
\begin{proof}
It is an immediate consequence of Lemma \ref{20110617:lem1}.
\end{proof}

According to Lemmas \ref{20110617:lem2} and \ref{20110617:lem3}, if $K(\partial)\in\Mat_{\ell\times\ell}(\mc F[\partial])$ is skewadjoint,
the restriction of $\langle\,\cdot|\cdot\,\rangle_K$ to $\mc Z\subset\mc F^\ell$
defines a symmetric bilinear form on $\mc Z$ with values in $\mc C$, which we denote by
$$
\langle\,\cdot|\cdot\,\rangle^0_K:=\langle\,\cdot|\cdot\,\rangle_K\big|_{\mc Z}:\,\,\mc Z\times\mc Z\to\mc C\,.
$$

\begin{lemma}\label{20110617:lem4}
Assuming that $K(\partial)\in\Mat_{\ell\times\ell}(\mc F[\partial])$ is a skewadjoint operator
and $P(\partial)\in\Mat_{\ell\times\ell}(\mc F[\partial])$ is such that $K(\partial)P(\partial)+P^*(\partial)K(\partial)=0$,
we have
$$
\langle P(\partial)F|G\rangle_K+\langle F|P(\partial)G\rangle_K=0
$$
for every $F,G\in\mc F^\ell$.
\end{lemma}
\begin{proof}
By equation \eqref{20110617:eq3}, we have
$$%\begin{equation}\label{20110617:eq4}
\begin{array}{l}
\displaystyle{
\langle P(\partial)F|G\rangle_K
} \\
\displaystyle{
=\sum_{i,j,k\in I}
(\lambda+\mu+\partial)^{-1}\big(K_{kj}(\mu)+K_{jk}(\lambda)\big)
\big(|_{\lambda=\partial}P_{ki}(\partial)F_i\big)\big(|_{\mu=\partial}G_j\big)
} \\
\displaystyle{
=\sum_{i,j,k\in I}
(\lambda+\mu+\partial)^{-1}\big(K_{kj}(\mu)+K_{jk}(\lambda+\partial)\big)P_{ki}(\lambda)
\big(|_{\lambda=\partial}F_i\big)\big(|_{\mu=\partial}G_j\big)
} \\
\displaystyle{
=\!\!\!\sum_{i,j,k\in I}\!\!\!
(\lambda+\mu+\partial)^{-1}\big(P_{ki}(\lambda)K_{kj}(\mu)
\!-\!
P^*_{jk}(\lambda+\mu)K_{ki}(\lambda)\big)
\big(|_{\lambda=\partial}F_i\big)\big(|_{\mu=\partial}G_j\big)
.}
\end{array}
$$%\end{equation}
In the last identity we used the assumption that $K(\partial)P(\partial)=-P^*(\partial)K(\partial)$.
Similarly,
$$%\begin{equation}\label{20110617:eq5}
\begin{array}{l}
\displaystyle{
\langle F|P(\partial)G\rangle_K
%} \\
%\displaystyle{
%=\sum_{i,j,k\in I}
%(\lambda+\mu+\partial)^{-1}\big(K_{ik}(\mu)+K_{ki}(\lambda)\big)
%\big(|_{\lambda=\partial}F_i\big)\big(|_{\mu=\partial}P_{kj}(\partial)G_j\big)
%} \\
%\displaystyle{
%=\sum_{i,j,k\in I}
%(\lambda+\mu+\partial)^{-1}\big(K_{ik}(\mu+\partial)P_{kj}(\mu)+K_{ki}(\lambda)P_{kj}(\mu)\big)
%\big(|_{\lambda=\partial}F_i\big)\big(|_{\mu=\partial}G_j\big)
=\sum_{i,j,k\in I}
(\lambda+\mu+\partial)^{-1}
} \\
\displaystyle{
\times\big(-P^*_{ik}(\mu+\partial)K_{kj}(\mu)+P_{kj}(\mu)K_{ki}(\lambda)\big)
\big(|_{\lambda=\partial}F_i\big)\big(|_{\mu=\partial}G_j\big)
\,.}
\end{array}
$$%\end{equation}
Combining these two equations, we get
\begin{equation}\label{20110617:eq6}
\begin{array}{l}
\displaystyle{
\langle P(\partial)F|G\rangle_K
+\langle F|P(\partial)G\rangle_K
} \\
\displaystyle{
=\sum_{i,j,k\in I}
(\lambda+\mu+\partial)^{-1}
\Big(
\big(P_{ki}(\lambda)-P^*_{ik}(\mu+\partial)\big)K_{kj}(\mu)
} \\
\displaystyle{
\,\,\,\,\,\,\,\,\,\,\,\,\,\,\,\,\,\,
+\big(P_{kj}(\mu)-P^*_{jk}(\lambda+\mu)\big)K_{ki}(\lambda)
\Big)
\big(|_{\lambda=\partial}F_i\big)\big(|_{\mu=\partial}G_j\big)
\,.}
\end{array}
\end{equation}
We next observe that the differential operator $P_{ki}(\lambda)-P^*_{ik}(\mu+\partial)$
lies in $(\lambda+\mu+\partial)\circ(\mc F[\lambda,\mu])[\partial]$,
i.e. it is of the form
$$
P_{ki}(\lambda)-P^*_{ik}(\mu+\partial)=(\lambda+\mu+\partial)\circ Q_{ki}(\lambda,\mu+\partial)\,,
$$
for some polynomial $Q_{ki}$.
Hence,
$$
(\lambda+\mu+\partial)^{-1}\big(P_{ki}(\lambda)-P^*_{ik}(\mu+\partial)\big)K_{kj}(\mu)\big(|_{\mu=\partial}G_j\big)
=Q_{ik}(\lambda,\partial)K_{kj}(\partial)G_j\,,
$$
which, after summing with respect to $j\in I$, becomes zero since, by assumption, $G\in\ker(K(\partial))$.
Similarly,
$$
(\lambda+\mu+\partial)^{-1}
\big(P_{kj}(\mu)-P^*_{jk}(\lambda+\mu)\big)
K_{ki}(\lambda)
\big(|_{\lambda=\partial}F_i\big)
=Q_{kj}(\mu,\partial)K_{ki}(\partial)F_i\,,
$$
which is zero after summing with respect to $i\in I$, since $F\in\ker(K(\partial))$.
Therefore the RHS of \eqref{20110617:eq6} is zero, proving the claim.
\end{proof}
\begin{proposition}\label{20110613:prop}
Assuming that $\mc F$ is a linearly closed differential field,
and that $K(\partial)\in\Mat_{\ell\times\ell}(\mc F[\partial])$
is a skewadjoint $\ell\times\ell$ matrix differential operator with invertible leading coefficient,
the $\mc C$-bilieanr form $\langle\,\cdot|\cdot\,\rangle^0_K:\,\mc Z\times\mc Z\to\mc C$
is nondegenerate.
\end{proposition}
\begin{proof}
Given $F\in\mc F^\ell$, consider the map $P_F:\,\mc F^\ell\to\mc F$ given by $G\mapsto P_F(G)=\langle F|G\rangle^0_K$.
Equation \eqref{20110613:eq2} can be rewritten
by saying that $P_F$ is a $1\times\ell$ matrix differential operator, of order less than or equal to $N-1$,
with entries
$$
(P_F)_{j}(\partial)
= \sum_{i\in I}\sum_{n=0}^N\sum_{m=0}^{n-1} \binom{n}{m}
(-\partial)^{n-1-m} \circ F_i K_{ij;n} \partial^m \,.
$$
Suppose now that $P_F(G)=\langle P|G\rangle^0_K=0$ for all $G\in\mc Z\subset\mc F^\ell$.
By Lemma \ref{20110612:lem2} we get that $P_F(\partial)=0$.
On the other hand, the (left) coefficient of $\partial^{N-1}$ in $(P_F)_j(\partial)$ is
$$
0=\sum_{i\in I}\sum_{m=0}^{N-1} \binom{N}{m}
(-1)^{N-1-m} F_i (K_N)_{ij}=\sum_{i\in I}F_i(K_N)_{ij}\,.
$$
Since, by assumption, $K_N\in\Mat_{\ell\times\ell}(\mc F)$ is invertible, we conclude that $F=0$.
\end{proof}

\subsection{Proof of Theorem \ref{20110601:thm}}\label{sec:4.2}

Recall from Lemma \ref{20110612:lemf} that $\mc H^{-1}_K=\mc Z^{-1}_K$
is isomorphic, as a $\mc C$-vector space, to $\mc Z=\ker\big(K(\partial)\big)$,
and, from Theorem \ref{110213:thm}, that $\dim_{\mc C}\mc Z=N\ell$.
By Corollary \ref{20110612:thm}, the $\mb Z$-graded Lie superalgebra $\mc H^\bullet_K(\mc V)$
is transitive, i.e. if $P\in\mc H^k_K,\,k\geq0$,
is such that $[P,\mc H^{-1}_K]=0$, then $P=0$.
Hence, due to transitivity, the representation of $\mc H^0$ on $\mc H^{-1}_K=\mc Z^{-1}_K$ is faithful.
Identifying $\mc Z^{-1}_K\simeq\mc Z$, we can therefore view
$\mc H^0_K$ as a subalgebra of the Lie algebra $gl(\mc Z)=gl_{N\ell}$.
Recall, from Theorem \ref{110213:thm} that $\mc H^0_K\simeq\mc A^0_K$
consists of elements of the form $Q=\big(\sum_j P^*_{ij}(\partial)u_j\big)_{i\in I}\in\mc V^\ell$,
where $P(\partial)=\big(P_{ij}(\partial)\big)_{i\in I}$ is an $\ell\times\ell$ matrix differential operator
of order at most $N-1$ solving equation \eqref{110226:eq5}.
Moreover, by \eqref{20110613:eq3}, the bracket of an element $Q\in\mc H^0_K$ as above and an element
$C\in\mc Z^{-1}_K=\mc H^{-1}_K\subset\mc V/\partial\mc V$, is given by
$$
[Q,C]
=\sum_{i,j\in I}\int \big(P^*_{ij}(\partial)u_j\big)\frac{\delta C}{\delta u_i}
=\sum_{i,j\in I}\int u_iP_{ij}(\partial)\frac{\delta C}{\delta u_j}\,.
$$
Hence, by the identification \eqref{20110614:eq1},
the corresponding action of $Q\in\mc H^0_K$ on $\mc Z\subset\mc F^\ell$ is simply given by
the standard action of the $\ell\times\ell$ matrix differential operator $P(\partial)$ on $\mc F^\ell$.
By Lemmas \ref{20110617:lem2} and \ref{20110617:lem3} and by Proposition \ref{20110613:prop},
$\langle\,\cdot|\cdot\,\rangle^0_K$ is a nondegenerate symmetric bilinear form on $\mc Z$,
and by Lemma \ref{20110617:lem4}
it is invariant with respect to this action of $Q\in\mc H^0_K$ on $\mc Z$.
Hence, the image of $\mc H^0_K$ via the above embedding $\mc H^0_K\to gl(\mc Z)$,
is a subalgebra of $so(\mc Z,\langle\,\cdot|\cdot\,\rangle^0_K)$.
Due to transitivity of the $\mb Z$-graded Lie superalgebra $\mc H^\bullet_K(\mc V)$, it
embeds in the full prolongation of the pair $\big(\mc Z,so(\mc Z,\langle\,\cdot|\cdot\,\rangle^0_K)\big)$,
which, by Proposition \ref{20110622:prop}, is isomorphic to $\tilde H(N\ell,S)$,
where $S$ is the $N\ell\times N\ell$ matrix of the bilinear form $\langle\,\cdot|\cdot\,\rangle^0_K$, in some basis.
By Theorem \ref{110213:thm}, $\dim_{\mc C}\mc H^k_{K}=\binom{N\ell}{k+2}$, which is equal to $\dim_{\mc C}\tilde H_k(N\ell,S)$.
We thus conclude that the $\mb Z$-graded Lie superalgebras $\mc H^\bullet_K(\mc V)$ and $\tilde H(N\ell,S)$ are isomorphic.
\begin{remark}
The same arguments as above show that,
without any assumption on the algebra of differential functions $\mc V$ and on the differential field $\mc F$
(with subfield of constants $\mc C$),
and for every Hamiltonian operator $K$ (not necessarily quasiconstant nor with invertible leading coefficient),
we have an injective homomorphism of $\mb Z$-graded Lie superalgebras
$\mc H^\bullet_K(\mc V)/\mc E\mc H^\bullet_K(\mc V)\to W(n)$,
where $n=\dim_{\mc C}(\mc H^{-1}_K)$.
\end{remark}

%%%%%%%%%%%%%%%%%%%%%%%%%%%%%%%%%%%%%%%%%%%%%%%%%%%%%%%%%%%%%%%%%%%%%%%%%%%%%%%%%%%%%%%%%%%%%%%%%%%%%%%%%%%%%%
%%%%%%%%%%%%%%% Sect 5 %%%%%%%%%%%%%%%%%%%%%%%%%%%%%%%%%%%%%%%%%%%%%%%%%%%%%%%%%%%%%%%%%%%%%%%%%%%%%%%%%%%%%%%
%%%%%%%%%%%%%%%%%%%%%%%%%%%%%%%%%%%%%%%%%%%%%%%%%%%%%%%%%%%%%%%%%%%%%%%%%%%%%%%%%%%%%%%%%%%%%%%%%%%%%%%%%%%%%%

\section{Translation invariant variational Poisson cohomology}
\label{sec:5}

In the previous sections we studied the variational Poisson cohomology $\tilde H^\bullet_K(\mc V)$
in the simplest case when the differential field of quasiconstants $\mc F\subset\mc V$ is linearly closed.
In this section we consider the other extreme case, often studied in literature --
the translation invariant case, when $\mc F=\mc C$.

\subsection{Upper bound of the dimension of the translation invariant variational Poisson cohomology}
\label{sec:5.1}

Let $\mc V$ be a normal algebra of differential functions, and assume that it is \emph{translation invariant},
i.e. the differential field $\mc F$ of quasiconstants coincides with the field $\mc C$ of constants.
Let $K(\partial)$ be an $\ell\times\ell$ matrix differential operator of order $N$,
with coefficients in $\Mat_{\ell\times\ell}(\mc C)$, and with invertible leading coefficient $K_N$.

For $k\geq-1$, denote by $\tilde{\mc H}^k$ the space of arrays $\big(P_{i_0,\dots,i_k}(\lambda_0,\dots,\lambda_k)\big)_{i_0,\dots,i_k\in I}$
with entries $P_{i_0,\dots,i_k}(\lambda_0,\dots,\lambda_k)\in\mc C[\lambda_0,\dots,\lambda_k]$,
of degree at most $N-1$ in each variable,
which are skewsymmetric with respect to simultaneous permutations of the indices $i_0,\dots,i_k$
and the variables $\lambda_0,\dots,\lambda_k$
(in the notation of \cite{DSK11}, $\tilde{\mc H}^k=\tilde\Omega^{k-1}_{0,0}$).
In particular, $\tilde{\mc H}^{-1}=\mc C$.
Note that, for $k\geq-1$, we have
\begin{equation}\label{20110618:eq2}
\dim_{\mc C}\tilde{\mc H}^k=\binom{N\ell}{k+1}\,.
\end{equation}

The long exact sequence \cite[eq.(11.4)]{DSK11} becomes (in the notation of the present paper):
\begin{equation}\label{20110618:eq1}
\begin{array}{l}
\displaystyle{
0\to\mc C
\stackrel{\beta_{-1}}{\longrightarrow}
\mc H^{-1}_K
\stackrel{\gamma_{-1}}{\longrightarrow}
\tilde{\mc H}^0
\stackrel{\alpha_0}{\longrightarrow}
\tilde{\mc H}^0
\stackrel{\beta_0}{\longrightarrow}
\dots
} \\
\displaystyle{
\dots
\stackrel{\gamma_{k-1}}{\longrightarrow}
\tilde{\mc H}^k
\stackrel{\alpha_{k}}{\longrightarrow}
\tilde{\mc H}^k
\stackrel{\beta_{k}}{\longrightarrow}
\mc H^{k}_K
\stackrel{\gamma_{k}}{\longrightarrow}
\tilde{\mc H}^{k+1}
\stackrel{\alpha_{k+1}}{\longrightarrow}
\tilde{\mc H}^{k+1}
\stackrel{\beta_{k+1}}{\longrightarrow}
\dots
}
\end{array}
\end{equation}

For every $k\geq-1$, we have $\dim_{\mc C}(\mc H^k_K)=\dim_{\mc C}(\ker\gamma_k)+\dim_{\mc C}(\im\gamma_k)$.
By exactness of the sequence \eqref{20110618:eq1},
we have that $\dim_{\mc C}(\im\gamma_k)=\dim_{\mc C}(\ker\alpha_{k+1})$,
and $\dim_{\mc C}(\ker\gamma_k)=\dim_{\mc C}(\im\beta_k)$.
Moreover, $\dim_{\mc C}(\im\beta_{-1})=1$ and, for $k\geq0$,
we have, again by exactness of \eqref{20110618:eq1}, that
$\dim_{\mc C}(\im\beta_k)=\dim_{\mc C}\tilde{\mc H}^k-\dim_{\mc C}(\ker\beta_k)
=\dim_{\mc C}\tilde{\mc H}^k-\dim_{\mc C}(\im\alpha_k)=\dim_{\mc C}(\ker\alpha_k)$.
Hence, using \eqref{20110618:eq2} we conclude that
\begin{equation}\label{20110620:eq2}
\dim_{\mc C}(\mc H^{-1}_K)=1+\dim_{\mc C}(\ker\alpha_{0})\leq N\ell+1\,,
\end{equation}
and, for $k\geq0$ (by the Tartaglia-Pascal triangle),
\begin{equation}\label{20110620:eq1}
\dim_{\mc C}(\mc H^k_K)=\dim_{\mc C}(\ker\alpha_k)+\dim_{\mc C}(\ker\alpha_{k+1})
%\leq\binom{N\ell}{k+1}+\binom{N\ell}{k+2}
\leq\binom{N\ell+1}{k+2}\,.
\end{equation}

Recalling equation \eqref{20110612:eq1}, we have
$\mc H^{-1}_K=\mc Z^{-1}_K=\big\{\tint f\in\mc V/\partial\mc V\,\big|\,K(\partial)\frac{\delta f}{\delta u}=0\big\}$.
By Lemma \ref{20110612:lemf}(b) we have
a surjective map $\frac{\delta}{\delta u}:\,\mc H^{-1}_K\to\ker\big(K(\partial)\big|_{\mc C^\ell}\big)$.
Recall that, if $\mc V$ is a normal algebra of differential functions,
we have $\ker\big(\frac{\delta}{\delta u}:\,\mc V\to\mc V^\ell\big)=\mc C+\partial\mc V$ \cite{BDSK09}.
It follows that $\ker\big(\frac{\delta}{\delta u}\big|_{\mc H^{-1}_K}\big)=\ker\big(\frac{\delta}{\delta u}\big|_{\mc V/\partial\mc V}\big)\simeq\mc C$.
Therefore,
$$
\mc H^{-1}_K
=
\tint\mc C\oplus\big\{\tint uA\,\big|\,A\in\ker(K_0)\subset\mc C^\ell\big\}\,,
$$
where, $u=(u_1,\dots,u_\ell)$, and $K_0=K(0)$ is the constant coefficient of the differential operator $K(\partial)$.
Hence,
\begin{equation}\label{20110629:eq4}
\dim_{\mc C}(\mc H^{-1}_K)
=1+\dim_{\mc C}(\ker K_0)
=1+\ell-\rk(K_0)\,.
\end{equation}
In conclusion, the inequality in \eqref{20110620:eq2} is a strict inequality
unless $K(\partial)$ has order 1 with $K_0=0$, i.e. $K(\partial)=S\partial$,
where $S\in\Mat_{\ell\times\ell}(\mc C)$ is a nondegenerate matrix.

\begin{remark}\label{20110629:rem1}
The map $\alpha_k:\,\tilde{\mc H}^k\to\tilde{\mc H}^k$ can be constructed as follows \cite{DSK11}.
Let $P=\big(P_{i_0,\dots,i_k}(\lambda_0,\dots,\lambda_k)\big)_{i_0,\dots,i_k\in I}$ be in $\tilde{\mc H}^k$,
i.e. $P_{i_0,\dots,i_k}(\lambda_0,\dots,\lambda_k)$ are
polynomials of degree at most $N-1$ in each variable $\lambda_i$ with coefficients in $\mc C$,
skewsymmetric with respect to simultaneous permutations in the indices $i_0,\dots,i_k$ 
and the variables $\lambda_0,\dots,\lambda_k$.
Then, there exist a unique element $\alpha_k(P):=R=\big(R_{i_0,\dots,i_k}(\lambda_0,\dots,\lambda_k)\big)_{i_0,\dots,i_k\in I}\in\tilde{\mc H}^k$
and a (unique) array 
$Q=\big(Q_{j,i_1,\dots,i_k}(\lambda_1,\dots,\lambda_k)\big)_{j,i_1,\dots,i_k\in I}$,
where $Q_{j,i_1,\dots,i_k}(\lambda_1,\dots,\lambda_k)$ are polynomials of degree at most $N-1$ in each variable,
with coefficients in $\mc C$, skewsymmetric with respect of simultaneous permutations of the indices $i_1,\dots,i_k$
and the variables $\lambda_1,\dots,\lambda_k$,
such that the following identity holds in $\mc C[\lambda_0,\dots,\lambda_k]$:
%
%is defined as the (unique) element solving the equation
%\begin{equation}\label{20110629:eq3}
%\begin{array}{c}
%\displaystyle{
%(\lambda_0+\dots+\lambda_k+\partial)P_{i_0,\dots,i_k}(\lambda_0,\dots,\lambda_k)
%=
%R_{i_0,\dots,i_k}(\lambda_0,\dots,\lambda_k)
%} \\
%\displaystyle{
%+ \sum_{\alpha=0}^k(-1)^\alpha
%\sum_{j\in I,\,n\in\mb Z_+}
%\frac{\partial Q_{i_0,\stackrel{\alpha}{\check{\dots}},i_k}(\lambda_0,\stackrel{\alpha}{\check{\dots}},\lambda_k)}{\partial u_j^{(n)}}(\lambda_\alpha+\partial)^n %K_{ji_\alpha}(\lambda_\alpha)
%\,,
%}
%\end{array}
%\end{equation}
%for some (non necessarily unique) element
%$Q=\big(Q_{i_1,\dots,i_k}(\lambda_1,\dots,\lambda_k)\big)_{i_1,\dots,i_k\in I}\in W^{var}_{k-1}$.
%%
%One can show (using the arguments in \cite{DSK11}), that $Q$ in \eqref{20110629:eq3} can be chosen (uniquely) of the form
%$$
%Q=
%\Big(\sum_{j\in I}u_jQ_{j,i_1,\dots,i_k}(\lambda_1,\dots,\lambda_k)\Big)_{i_1,\dots,i_k\in I}\,,
%$$
%where now $Q_{j,i_1,\dots,i_k}(\lambda_1,\dots,\lambda_k)$ are polynomials of degree at most $N-1$ in each variable,
%with coefficients in $\mc C$, skewsymmetric with respect of simultaneous permutations of the indices $i_1,\dots,i_k$
%and the variables $\lambda_1,\dots,\lambda_k$.
%%
%With this choice, equation \eqref{20110629:eq3} reads:
%
\begin{equation}\label{20110629:eq5}
\begin{array}{c}
\displaystyle{
(\lambda_0+\dots+\lambda_k)P_{i_0,\dots,i_k}(\lambda_0,\dots,\lambda_k)
=
R_{i_0,\dots,i_k}(\lambda_0,\dots,\lambda_k)
} \\
\displaystyle{
+ \sum_{\alpha=0}^k(-1)^\alpha
\sum_{j\in I}
Q_{j,i_0,\stackrel{\alpha}{\check{\dots}},i_k}(\lambda_0,\stackrel{\alpha}{\check{\dots}},\lambda_k)
K_{ji_\alpha}(\lambda_\alpha)
\,.
}
\end{array}
\end{equation}
Hence, $\ker(\alpha_k)$ is in bijection with the space $\Sigma_k$ of arrays $Q$ as above,
satisfying the condition:
$$
\sum_{\alpha=0}^k(-1)^\alpha
\sum_{j\in I}
Q_{j,i_0,\stackrel{\alpha}{\check{\dots}},i_k}(\lambda_0,\stackrel{\alpha}{\check{\dots}},\lambda_k)K_{ji_\alpha}(\lambda_\alpha)
\in(\lambda_0+\dots+\lambda_k)\mc C[\lambda_0,\dots,\lambda_k]\,.
$$
For example, $\Sigma_0=\big\{Q\in\mc C^\ell\,\big|\,K_0^TQ=0\big\}$,
hence its dimension equals $\dim_{\mc C}(\ker\alpha_0)=\dim(\ker K_0)=\ell-\rk(K_0)$ (in accordance with \eqref{20110629:eq4}).
Furthermore, $\Sigma_1$ consists of polynomials $Q(\lambda)$ with coefficients in $\Mat_{\ell\times\ell}(\mc C)$,
of degree at most $N-1$, such that
$$
K^T(-\lambda)Q(\lambda)=Q^T(-\lambda)K(\lambda)\,.
$$
\end{remark}
\begin{remark}
It is clear from Remark \ref{20110629:rem1} that, while in the linearly closed case,
the Lie superalgebra $\mc H^\bullet_K(\mc V)$ depends only on $\ell$ and the order $N$ of $K(\partial)$,
in the translation invariant case $\mc F=\mc C$ the dimension of $\mc H^\bullet_K(\mc V)$
depends essentially on the operator $K(\partial)$.
Hence, in this sense, the choice of an algebra $\mc V$ over a linearly closed differential field $\mc F$
seems to be a more natural one.
This is the key message of the paper.
%
% (which the worldwide community in integrable systems should take into account).
\end{remark}

In the next section we study in more detail the variational Poisson cohomology $\mc H^k_K$,
and its $\mb Z$-graded Lie superalgebra structure, for a ``hydrdynamic type'' Hamiltonian operator,
i.e. for $K(\partial)=S\partial$, where $S\in\Mat_{\ell\times\ell}(\mc C)$ is nondegenerate and symmetric.

\subsection{Translation invariant variational Poisson cohomology for $K=S\partial$}
\label{sec:5.2}

As in the previous section, let $\mc V$ be a translation invariant normal algebra of differential functions,
with field of constants $\mc C$ (which coincides with the field of quasiconstants).
Let $S\in\Mat_{\ell\times\ell}(\mc C)$ be nondegenerate and symmetric,
and consider the Hamiltonian operator $K(\partial)=S\partial$.

For $k\geq-1$, we denote by $\Lambda^{k+1}$ the space of skewsymmetric $(k+1)$-linear forms on $\mc C^\ell$,
i.e. the space of arrays
$B=\big(b_{i_0,\dots,i_k}\big)_{i_0,\dots,i_k\in I}$,
totally skewsymmetric with respect to permutations of the indices $i_0,\dots,i_k$.
For $k\geq0$, we also denote by $\Lambda_S^{k+1}$ the space of arrays of the form
$A=\big(a_{j,i_1,\dots,i_k}\big)_{j,i_1,\dots,i_k\in I}$,
which are skewsymmetric with respect to permutations of the indices $i_1,\dots,i_k$,
and which satisfy the equation
$$
\sum_{j\in I}s_{i_0,j}a_{j,i_1,i_2\dots,i_{k}}=
-\sum_{j\in I}a_{j,i_0,i_2,\dots,i_{k}}s_{j,i_1}\,.
$$
Clearly, $\dim_{\mc C}(\Lambda^{k+1}_S)=\dim_{\mc C}(\Lambda^{k+1})=\binom{\ell}{k+1}$ for every $k\geq-1$.
For example, $\Lambda^0=\mc C$, $\Lambda^1_S=\Lambda^1=\mc C^\ell$,
$\Lambda^2$ is the space of skewsymmetric $\ell\times\ell$ matrices over $\mc C$, and
$$
\Lambda^2_S=\big\{A\in\Mat{}_{\ell\times\ell}(\mc C)\,\big|\,A^TS+SA=0\big\}=so(\ell,S)\,.
$$
Given $A=\big(a_{j,i_0,\dots,i_k}\big)_{j,i_0,\dots,i_k\in I}\in\Lambda_S^{k+2}$,
we denote
$$
u A=\big(\sum_{j\in I}u_ja_{j,i_0,\dots,i_k}\big)_{i_0,\dots,i_k\in I}\in W^{var}_{k}\,.
$$
Let $\mc A^\bullet=\bigoplus_{k=-1}^\infty\mc A^k$, where
$$
\mc A^{k}
=\Lambda^{k+1}\oplus
\big\{uA\,\big|\,A\in\Lambda^{k+2}_S\big\}\subset W^{var}_{k}\,,\,\,k\geq-1.
$$

\begin{theorem}\label{20110622:thm}
Let $\mc V$ be trnslation invariant normal algebra of differential functions,
and let $K(\partial)=S\partial$, where $S$ is a symmetric nondegenerate $\ell\times\ell$
matrix over $\mc C$. Then:
\begin{enumerate}[(a)]
\item
$\mc A^\bullet$ is a subalgebra of the $\mb Z$-graded Lie superalgebra $\mc Z^\bullet_K(\mc V)$,
complementary to the ideal $\mc B^\bullet_K(\mc V)$.
In particular, we have the following decomposition of $\mc Z^k_K$ in a direct sum of vector spaces over $\mc C$:
$$
\mc Z^k_K=\mc A^k\oplus\mc B^k_K\,.
$$
\item
We have an isomorphism of $\mb Z$-graded Lie superalgebras (cf. Secction \ref{sec:before.2}):
$$
\mc H^\bullet_K(\mc V)=\mc A^\bullet\simeq\tilde H(\ell+1,\tilde S)\,,
$$
where
$\tilde S$ is the $(\ell+1)\times(\ell+1)$ matrix obtained from $S$ by adding a zero row and column.
In particular, $\dim_\mc C(\mc H^k_K)=\binom{l+1}{k+2}$.
\end{enumerate}
\end{theorem}
\begin{proof}
For $B\in\Lambda^{k+1}$, we obviously have $\delta_KB=0$.
Moreover, it is immediate to check, using the formula \eqref{20110611:eq2} for $\delta_K$,
that, if $A\in\Lambda^{k+2}_S$, then $\delta_K(uA)=0$.
Hence, $\mc A^k\subset\mc Z^k_K$ for every $k\geq-1$.
Next, we compute the box product \eqref{100418:eq2} between two elements of $\mc A^\bullet$.
Let $B\oplus uA\in\Lambda^{h+1}\oplus u\Lambda^{h+2}_S=\mc A^h$,
and $D\oplus uC\in\Lambda^{k-h+1}\oplus u\Lambda^{k-h+2}_S=\mc A^{k-h}$.
We have $B\Box D=0$, $uA\Box D=0$,
moreover,
$B\Box uC\in\Lambda^{k+1}\subset\mc A^k$ and $uA\Box uC\in u\Lambda^{k+2}_S\subset\mc A$
are given by
\begin{equation}\label{20110627:eq3}
\begin{array}{l}
\displaystyle{
{(B\Box uC)}_{i_0,\dots,i_k}
= \sum_{\sigma\in\mc S_{h,k}}
\sign(\sigma)
\sum_{j\in I}
b_{j,i_{\sigma(k-h+1)},\dots,i_{\sigma(k)}}
c_{j,i_{\sigma(0)},\dots,i_{\sigma(k-h)}}
\,,
} \\
\displaystyle{
{(uA\Box uC)}_{i_0,\dots,i_k}
= \sum_{\sigma\in\mc S_{h,k}}
\sign(\sigma)
\sum_{i,j\in I}
u_i a_{i,j,i_{\sigma(k-h+1)},\dots,i_{\sigma(k)}}
c_{j,i_{\sigma(0)},\dots,i_{\sigma(k-h)}}
\,.
}
\end{array}
\end{equation}
We thus conclude that $\mc A^\bullet=\bigoplus_{k\geq-1}\mc A^k$ is a subalgebra
of the $\mb Z$-graded Lie superalgebra $\mc Z^\bullet(\mc V)\subset W^{\var}(\Pi\mc V)$.

Since $\mc S_{-1,k+1}=\emptyset$, we have that $\mc A^{-1}\Box\mc A^\bullet=0$.
Moreover, $\mc S_{-1,k+1}=\{1\}$.
Hence, for $d\oplus uC\in\mc C\oplus u\mc C^\ell=\mc A^{-1}$
and $B\oplus uA\in\Lambda^{k+1}\oplus u\Lambda^{k+2}_S=\mc A^k$,
we have
$$
[B\oplus uA,d\oplus uC]=B\Box(uC)\oplus(uA\Box uC)\in\Lambda^{k}\oplus u\Lambda^{k+1}_S=\mc A^{k-1}\,,
$$
with entries
\begin{equation}\label{20110627:eq4}
\begin{array}{l}
\displaystyle{
[B,uC]_{i_1,\dots,i_k}
=
(B\Box uC)_{i_1,\dots,i_k}
=
\sum_{j\in I}
b_{j,i_1,\dots,i_k}
c_{j}
\,,
} \\
\displaystyle{
[uA,uC]_{i_1,\dots,i_k}
=
(uA\Box uC)_{i_1,\dots,i_k}
=
\sum_{i,j\in I}
u_i a_{i,j,i_1,\dots,i_k}
c_{j}
\,.
}
\end{array}
\end{equation}

It is clear, from formula \eqref{20110627:eq4}, that $[B\oplus uA,uC]=0$ for every $C\in\mc C^\ell$
if and only if $A=0$ and $B=0$.
Hence $\mc A^\bullet$ is a transitive $\mb Z$-graded Lie superalgebra.

Since $[\mc B^k_K,\mc Z^{-1}_K]=0$, it follows, in particular, that $\mc A^k\cap\mc B^k_K=0$ for every $k\geq-1$.
Hence $\mc A^k$ coincides with its image in $\mc H^k_K(\mc V)$,
and $\mc A^\bullet$ can be viewed as a subalgebra of the $\mb Z$-graded Lie superalgebra $\mc H^\bullet_K(\mc V)$.
Therefore (by the Tartaglia-Pascal triangle) $\dim_{\mc C}\mc H^k_K\geq\dim_{\mc C}\mc A^k=\binom{\ell+1}{k+2}$.
Since, by \eqref{20110620:eq1}, $\dim_{\mc C}\mc H^k_K\leq\binom{\ell+1}{k+2}$,
we conclude that all these inequalities are equalities,
and that $\mc H^\bullet(\mc V)\simeq\mc A^\bullet$ are isomorphic $\mb Z$-graded Lie superalgebras.

To conclude, in view of Proposition \ref{20110622:prop},
we need to prove that $\mc A^\bullet$ is the full prolongation of the pair
$(\mc C^{\ell+1},so(\ell+1,\tilde S)$,
where $\tilde S$ is the $(\ell+1)\times(\ell+1)$ matrix obtained adding a zero row and column to $S$.
We have
$\mc C^{\ell+1}=\mc C\oplus\mc C^\ell$, and
$$
so(\ell+1,\tilde S)
=\Big\{\left(\begin{array}{cc} 0 & B^T \\ 0 & A \end{array}\right)
\,\Big|\, B\in\mc C^{\ell}\,,\,\, A\in so(\ell,S)\Big\}
\simeq\mc C^\ell\oplus so(\ell,S)
\,,
$$
with the Lie bracket of $B\oplus A\in\mc C^\ell\oplus so(\ell,S)$ and $d\oplus C\in\mc C\oplus\mc C^\ell$ given by
\begin{equation}\label{20110627:eq1}
[B+A,d+C]=B\cdot C\oplus AC\in\mc C\oplus\mc C^\ell\,.
\end{equation}

By definition, we have $\mc A^0=\Lambda^1\oplus u\Lambda^2_S=\mc C^\ell\oplus u\cdot so(\ell,S)$,
and the action of
$B\oplus uA\in \mc C^\ell\oplus u\cdot so(\ell,S)$
on $d\oplus uC\in\mc C\oplus u\mc C^\ell=\mc A^{-1}$, given by \eqref{20110627:eq4},
is $[B\oplus uA,d\oplus uC]_i=B\cdot C\oplus u AC$.
Namely, in view of \eqref{20110627:eq1}, it is induced by the natural action of
$so(\ell+1,\tilde S)\simeq\mc C^\ell\oplus so(\ell,S)$ on $\mc C\oplus\mc C^\ell$.
Hence, $\mc A^{-1}\oplus\mc A^0\simeq(\mc C\oplus\mc C^\ell)\oplus(\mc C^\ell\oplus so(n,S))$.
Since $\mc A^\bullet$ is a transitive $\mb Z$-graded Lie superalgebra,
it is a subalgebra of the full prolongation of $(\mc C^{\ell+1},so(\ell+1,\tilde S)$.

On the other hand, by Proposition \ref{20110622:prop} the full prolongation of $(\mc C^{\ell+1},so(\ell+1,\tilde S)$
is isomorphic to $\tilde H(\ell+1,\tilde S)$,
and $\dim_{\mc C}\tilde H(\ell+1,\tilde S)=2^{\ell+1}-1=\sum_{k\geq-1}\dim_{\mc C}\mc A^k$.
Hence, $\mc A^\bullet$ must be isomorphic to $\tilde H(\ell+1,\tilde S)$, as we wanted.
\end{proof}
\begin{corollary}\label{20110628:cor}
Under the assumptions of Theorem \ref{20110622:thm},
the essential variational cohomology $\mc E\mc H^\bullet_K(\mc V)$ is zero.
\end{corollary}
\begin{proof}
It immediately follows from the transitivity of the $\mb Z$-graded Lie superalgebra $\mc H^\bullet_K(\mc V)$.
\end{proof}

\begin{remark}\label{20110628:rem2}
If $S$ is a nondegenerate, but not necessarily symmetric, $\ell\times\ell$ matrix,
we still have an isomorphism of vector spaces $\mc H^k_K\simeq\mc A^k$,
but $\mc H^\bullet_K(\mc V)$ is not, in general, a Lie superalgebra.
\end{remark}
\begin{remark}\label{20110628:rem3}
The description of $\mc H^\bullet_K(\mc V)$, as a vector space, for $K=S\partial$ with $S$ symmetric nondegenerate matrix over $\mc C$,
agrees with the results of S.-Q. Liu and Y. Zhang \cite{LZ11,LZ11pr}.
\end{remark}
%

%%%%%%%%%%%%%%%%%%%%%%%%%%%%%%%%%%%%%%%%%%%%%%%%%%%%%%%%%%%%%%%%%%%%%%%%%%%%%%%%%%%%%%%%%%%%%%%%%%%%%%%%%%%%%%
%%%%%%%%%%%%%%% Bibliography %%%%%%%%%%%%%%%%%%%%%%%%%%%%%%%%%%%%%%%%%%%%%%%%%%%%%%%%%%%%%%%%%%%%%%%%%%%%%%%%%
%%%%%%%%%%%%%%%%%%%%%%%%%%%%%%%%%%%%%%%%%%%%%%%%%%%%%%%%%%%%%%%%%%%%%%%%%%%%%%%%%%%%%%%%%%%%%%%%%%%%%%%%%%%%%%

% Non-BibTeX users please use

\end{document}